\documentclass[11pt]{article}
\usepackage{amsfonts}
\usepackage{amsmath, amsthm, amssymb, bbm}
\usepackage{xcolor,multirow,graphicx,enumerate}
\usepackage{float}
\usepackage[title]{appendix}
\usepackage{natbib}
\usepackage{subcaption}
\usepackage{algorithm}
\usepackage{algorithmicx} 
\usepackage{algpseudocode}
\usepackage[breaklinks]{hyperref}

\newtheorem{assumption}{Assumption}
\newtheorem{remark}{Remark}
\newtheorem*{condition}{Regularity condition}
\newtheorem{theorem}{Theorem}

\newtheorem{lemma}{Lemma}
\newtheorem{corollary}{Corollary}

\newcommand{\plim}{\operatorname*{plim}}

\newcommand{\fct}{h}
\newcommand{\FCT}{g}

\newcommand{\argmin}{\operatorname*{argmin}}
\newcommand{\vect}{\operatorname*{vec}}

\newcommand{\norm}[1]{\left\lVert#1\right\rVert}

\usepackage{mathtools}
\DeclarePairedDelimiter\ceil{\lceil}{\rceil}
\DeclarePairedDelimiter\floor{\lfloor}{\rfloor}

\begin{document}

\title{\bf Linear Panel Regressions with \\ Two-Way Unobserved Heterogeneity\thanks{%
We are grateful for feedback from conference participants at the 27th International Panel
Data Conference, and the 2021 Bristol Econometric Study Group. We also thank
Eric Auerbach,	St\'ephane Bonhomme, Arturas Juodis, Liangjun Su, 
	Joakim Westerlund,
and Andrei Zeleneev for useful comments and discussions.
This research was
supported by the Economic and Social Research Council through the ESRC Centre for
Microdata Methods and Practice grant RES-589-28-0001, and by  the
European Research Council grants ERC-2014-CoG-646917-ROMIA and
ERC-2018-CoG-819086-PANEDA. 
}
}
\author{\setcounter{footnote}{2}
   Hugo Freeman%
   \footnote{
                   Department of Economics,
                   University College London,
                   Gower Street,
                   London WC1E~6BT,
                   UK.
                   Email:~{\tt hugo.freeman.16@ucl.ac.uk}
                   }
   \and Martin Weidner%
   \footnote{Department of Economics and Nuffield College, University of Oxford,
   Manor Road, Oxford OX1~3UQ, and Institute for Fiscal Studies, London, UK.
                   Email:~{\tt martin.weidner@economics.ox.ac.uk}
                   }
                   }
\date{\today}

\maketitle

\begin{abstract}
\noindent
We study linear panel regression models in which the
unobserved error term is an unknown smooth function of  two-way unobserved fixed effects. 
In standard additive or interactive fixed effect models the individual specific and 
time specific effects are assumed to enter with a known functional form (additive
or multiplicative).
In this paper, we allow for this functional form to be more general and unknown.
We discuss two different estimation approaches that allow consistent estimation of 
the regression parameters in this setting as the number of individuals and the number of time periods
grow to infinity.
The first approach uses the interactive fixed effect estimator in \cite{Bai2009}, which is still applicable here,
as long as the number of factors in the estimation grows asymptotically. The second 
approach first discretizes the two-way unobserved heterogeneity 
(similar to what \citealt{bonhomme2017discretizing} are doing for one-way heterogeneity)
and then estimates a simple linear fixed effect model with additive two-way grouped fixed effects. For both estimation methods we obtain asymptotic convergence results, perform Monte Carlo simulations,
and employ the estimators in an empirical application to UK house price data.

\end{abstract}

\newpage

\section{Introduction}

We consider the following panel data model for $i = 1,\dots,N$ cross-sectional units, and $t = 1,\dots,T$ time periods,
\begin{align}\label{model:general}
    Y_{it} &= X_{it}^\prime \, \beta + u_{it} ,
    &
    u_{it} &=  \fct(\alpha_i,\gamma_t) + \varepsilon_{it}, 
\end{align}
where $Y_{it}$ is an observed dependent variable, $X_{it} = (X_{it,1},\ldots,X_{it,K})'$ is a $K$-vector of observed explanatory variables, 
and  $u_{it}$ is an unobserved error term.
Within the unobserved error term, we have an unknown real-valued function 
$\fct(\cdot,\cdot)$ that depends on the (vector-valued) unobserved fixed effects $\alpha_i \in \mathbb{R}^{d_\alpha}$
 and $\gamma_t \in \mathbb{R}^{d_\gamma}$, which are allowed to be arbitrarily correlated with the observed regressors $X_{it}$,
 while $\varepsilon_{it}$ is a mean-zero error term that is uncorrelated with $X_{it}$. Our focus is on estimation of and
 inference on the parameter $\beta \in \mathbb{R}^K$ --- 
   the regression coefficient of $X_{it}$ on $Y_{it}$ when properly controlling for the unobserved $\alpha_i$ and $\gamma_t$.
 
The key model restrictions in \eqref{model:general} are the linearity in $X_{it}$ as well as the
additive separability between $X_{it}^\prime \, \beta$
and  $u_{it}$.
If the unobserved error term $u_{it}$ is of the more general form $u_{it} = \FCT(\alpha_i,\gamma_t,\xi_{it})$,
for some idiosyncratic errors $\xi_{it}$ that are identically distributed across $i$ and over $t$,  and independent of the covariates $X_{it}$, then under appropriate 
regularity conditions we can define
$\fct(\alpha_i,\gamma_t) = \mathbb{E}\left[ u_{it} \, \big| \, \alpha_i, \gamma_t \right]$
 and $\varepsilon_{it} = u_{it} - \fct(\alpha_i,\gamma_t)$
 to again obtain model~\eqref{model:general}. 
  The additive
 separability between $\fct(\alpha_i,\gamma_t)$ and $\varepsilon_{it}$
 is therefore not strictly required.
However, throughout this paper we take the representation of the model in \eqref{model:general} as the starting point for our analysis.

Analogous to the singular value decomposition of a {\it matrix}, 
there exists, under weak regularity conditions, 
the singular value decomposition of a {\it function} $\fct: \mathbb{R}^{d_\alpha} \times \mathbb{R}^{d_\gamma} \rightarrow \mathbb{R}$, which reads
\begin{align}\label{functionSVD}
     \fct(\alpha,\gamma) =   \sum_{r=1}^{\infty} \sigma_r \, \varphi_r(\alpha) \, \psi_r(\gamma) ,
\end{align}
for some functional singular values $\sigma_r>0$,
and appropriate normalized
functions  $\varphi_r :  \mathbb{R}^{d_\alpha}  \rightarrow \mathbb{R}$ and 
$\psi_r :  \mathbb{R}^{d_\gamma}  \rightarrow \mathbb{R}$, $r \in \{1,2,3,\ldots\}$.
Equation \eqref{functionSVD}  allows us to rewrite  model
 \eqref{model:general} as
\begin{align}\label{model:rewritten}
    Y_{it} &= X_{it}^\prime \, \beta +   \sum_{r=1}^{\infty} \lambda_{ir} \, f_{tr}  + \varepsilon_{it}, 
\end{align} 
with $\lambda_{ir}  := \sigma_r \, \varphi_r(\alpha_i)$ and $f_{tr}:= \psi_r(\gamma_t)$. Thus, our model can be viewed as a linear panel
regression model with unobserved ``factor structure'' or ``interactive fixed effects'', but where the number of factors $f_{tr}$
and corresponding factor loadings $\lambda_{ir}$ is infinite.
The same rewriting of a function $\fct(\alpha_i,\gamma_t)$ 
by an infinite sum $\sum_{r=1}^{\infty} \lambda_{ir} \, f_{tr}$ 
is used in \cite{menzel2017bootstrap}, but for a different model,
and with the goal of analyzing the bootstrap for multidimensional data.

Within a panel regression context, most of the existing literature
assumes that the number of  unobserved factors is finite, which, from our perspective, corresponds to a truncation of the infinite sequence of factors
in \eqref{model:rewritten}, that gives
\begin{align}\label{model:truncated}
    Y_{it} &= X_{it}^\prime \, \beta +   \sum_{r=1}^{R} \lambda_{ir} \, f_{tr}  +  e_{it},
\end{align} 
where $e_{it}:= \varepsilon_{it} +   \sum_{r=R+1}^{\infty} \lambda_{ir} \, f_{tr}$.  The interactive fixed effect
model in \eqref{model:truncated} is one possible approximation
of the model \eqref{model:general} that we explore in this paper,
and we will show that this approximation can be used to estimate 
$\beta$ consistently. However, we also explore another approximation
of $\fct(\alpha_i,\gamma_t)$ using two-way grouped fixed effects,
see Section~\ref{subsec:GroupFEidea} below, and we also derive
convergence rate results for the resulting  
grouped fixed effect estimator. Other approximation methods
for $\fct(\alpha_i,\gamma_t)$ are also conceivable, but are not 
explored in this paper.\footnote{
For example, to justify \eqref{functionSVD} we rely on the 
paper by \cite{griebel2014approximation}, which also discusses
the alternative ``sparse grid'' approximation. In our context,
the sparse grid approximation would correspond to
replacing  $\sum_{r=1}^{R} \lambda_{ir} \, f_{tr}$
by $ \sum_{r,q=1}^{R} \gamma_{rq} \, \lambda_{ir} \, f_{tq}$,
with some sparsity condition on the matrix $\gamma=(\gamma_{rq})$.
}

For datasets with both $N$ and $T$ large, the two currently dominant estimation methods 
for the panel regression model in \eqref{model:truncated} are the common correlated effect (CCE) estimator
of \cite{Pesaran2006} and the least-squares (LS) estimator (also called quasi maximum likelihood estimator) in \cite{Bai2009}.
Since those original papers by Pesaran and Bai, a large literature has emerged that has extended the CCE and LS estimation methods, and 
has analyzed the properties of those estimators in more general settings ---
see \cite{chudik2013large},
\cite{bai2016econometric},
and \cite{karabiyik2019econometric} for recent surveys.
We follow that literature here by also considering panels with both $N$ and $T$ large, that is, for our asymptotic results we consider 
$N,T \rightarrow \infty$.\footnote{
There is of course also work on model \eqref{model:truncated} in the context of short $T$ panels, for example,
 \cite{HoltzEakin-Newey-Rosen1988},
 \cite{AhnLeeSchmidt2001,AhnLeeSchmidt2013},
 \cite{sarafidis2009impact}
 \cite{juodis2018fixed,juodis2022linear},
 \cite{westerlund2019cce},
}

The ``conventional'' interactive fixed effect model in \eqref{model:truncated} is a special case of our model \eqref{model:general},
with $\alpha_i = \lambda_i = (\lambda_{i1},\ldots,\lambda_{iR})'$,
$\gamma_t = f_t = (f_{t1},\ldots,f_{tR})'$, and $\fct(\alpha_i,\gamma_t) = \lambda_i' f_t$. The key question that we ask in this paper
is what happens when the multiplicative factor structure $ \lambda_i' f_t$ is replaced by a more general non-linear 
factor structure $\fct(\alpha_i,\gamma_t)$. However, we do maintain all other assumptions of model \eqref{model:truncated}, in particular,
the homogenous regression coefficient $\beta$, and the additive separability between $ X_{it}^\prime \, \beta $ and the unobserved  error.

The main challenge that we need to tackle when considering this extension is that, if the data generating process is given by \eqref{model:general}, then 
the error term $e_{it}$ in \eqref{model:truncated} will generally be correlated with $X_{it}$, because $e_{it}$ contains
the truncated part $\sum_{r=R+1}^{\infty} \lambda_{ir} \, f_{tr}$ of the infinite factor structure,\footnote{
Notice that the majority of these truncated factors will be
``weak'', see  Onatski~\cite*{Onatski2010,Onatski2012} and
Chudik, Pesaran and Tosetti~\cite*{ChudikPesaranTosetti2011}
for the distinction between ``strong'' and ``weak'' factors.
}
and $\lambda_{ir}  = \varphi_r(\alpha_i)$ and $f_{tr}= \psi_r(\gamma_t)$ are functions of $\alpha_i$ and $\gamma_t$, which can be correlated with $X_{it}$. Once $e_{it}$ is correlated
with $X_{it}$ in this way, then the existing results for the CCE and the LS estimator are not applicable anymore.
The currently known results on the CCE and LS estimator in the presence of an infinite number of factors (e.g.\ \citealt{pesaran2011large}, \citealt{chudik2011weak}, and \citealt{westerlund2013estimation})
require that the ``unaccounted'' factors $ \sum_{r=R+1}^{\infty} \lambda_{ir} \, f_{tr}$ are uncorrelated
 with the regressors, so that they can be considered part of the error term $e_{it}$ without generating an endogeneity problem.

For the case that $X_{it}$ and $e_{it}$ are correlated, 
there exist instrumental variable (IV) generalizations of both the CCE and LS method (e.g. \citealt{harding2011least}, \citealt{LeeMoonWeidner2012}, \citealt{robertson2015iv}, \citealt{MoonShumWeidner2018}, and \citealt{norkute2021instrumental}),
but those require observed instruments $Z_{it}$ that are uncorrelated with $e_{it}$. We do not explore instrumental variable approaches
in this paper.

The two main theoretical contributions of our paper are as follows:
Firstly, we formally show that the LS estimator of \cite{Bai2009} can still provide consistent estimates of $\beta$ in
model~\eqref{model:general}, as long as the number of factors $R=R_{NT}$ used in estimation grows to infinity jointly with $N$ and $T$.  
 Secondly, we
     suggest an alternative estimator for $\beta$, which we denote the  {\it two-way group fixed-effect estimator}
     (generalizing ideas in \citealt{bonhomme2017discretizing} on the discretization of one-way heterogeneity), and we provide conditions
     under which this new estimator is $\sqrt{NT}$-consistent  as $N,T \rightarrow \infty$.
In addition, 
we also suggest inference procedures using both of these estimators, but
we do not formally derive inference results in this paper. Instead, 
we study the properties of our suggested confidence intervals
in Monte Carlo simulations. We also apply the estimators to an empirical
application on UK house price data.

When employing the LS estimator 
with factors from \cite{Bai2009} to
model \eqref{model:general}, we are
effectively estimating a misspecified model
--- the DGP is given by \eqref{model:general},
but the estimating equation by \eqref{model:truncated}. 
\cite{galvao2014estimation}
and 
\cite{juodis2020shock}
have recently studied
linear panel regression models
with additive fixed effects under misspecification.
We consider interactive fixed effects for
estimation here,
and the type of misspecification we allow
for is more restrictive. We therefore
do not have to introduce any pseudo-true parameter,
but we find that the LS estimator is still
consistent for the true value of $\beta$ under 
our assumptions.

It also natural to ask if
our non-linear model $\fct(\alpha_i,\gamma_t)$ is truly necessary,
and also if there is a way to test whether a more standard additive or multiplicative error component structure would be 
  sufficient to capture unobserved heterogeneity. 
For example, 
\cite{kapetanios2019testing}  provide a test for
 whether the multiplicative error component structure is necessary or
 whether a simpler two-way fixed effect estimator would be sufficient. 
In many applications they find
 evidence that the standard two-way fixed effect should work well without the need for interactive fixed-effects. 
However, we do not pursue such a testing approach here, because if the main
goal is inference on $\beta$, then size distortions due to pre-testing quickly 
become a  concern (see e.g.\ \citealt{guggenberger2010impact}). 
Instead, our recommendation for applied researcher is to report
two-way fixed effect estimates jointly with factor augmented estimates
and grouped fixed effect estimates in one table that is then subjected to human interpretation.

 In related work, allowing for the number of factors to grow with sample size has been considered in \cite{li2017determining}, where they explicitly detail a factor model with the number of factors growing with sample size. The difference to this paper is our model admits an infinite number of factors even in small samples
and considers finite factor estimation as an approximation to the true data generating process.

There also exist other work on non-linear generalizations of the
interactive fixed effect and factor model specification.
 \cite{zeleneev2020identification} considers 
 the same model
  \eqref{model:general} in the context of network 
 data, but in his baseline discussion, the outcome $Y_{ij}$ (instead of $Y_{it}$ here) is symmetric in $i$ and $j$. The main difference to our 
 work, however, is that Zeleneev estimates the model based on a strategy that
 identifies agents with similar fixed effect values based on the distribution of their outcomes.
His estimation method is accordingly also 
completely different to ours.

\cite{bodelet2020nonparametric} also consider non-linear functions in place of the standard linear factor model. In our notation, their model assumes
a series of smooth univariate functions 
of the form $\sum_{q=1}^Q\fct_{iq}(\gamma_{tq})$ 
for unobserved heterogeneity. 
Their approach models individual specific responses to structural shocks but is different to our approach, which uses a homogeneous bivariate function. Therefore, their approach allows for discontinuities across how individual effects are modelled whereas our assumption is more restrictive since variation across individuals, via $\alpha_i$, must be smooth.

Other papers on unobserved two-way heterogeneity in panel or network models 
either make more parametric assumptions (e.g.~\citealt{graham2017econometric}, \citealt{dzemski2019empirical}, \citealt{chen2020nonlinear}),
or employ stochastic block or graphon models
(e.g.~\citealt{holland1983stochastic}, \citealt{wolfe2013nonparametric}, \citealt{gao2015rate}, \citealt{auerbach2019identification}),
and are therefore less closely related to our paper.

There are also recent papers that use matrix completion methods for
the purpose of  treatment effect estimation in panel models with
two-way heterogeneity, e.g. \cite{athey17} and \cite{amjad18}, \cite{chlz20},
and \cite{fernandez2021low}. Those papers do not require the
additive separability between the regressors and error term in \eqref{model:general}, but as a result they also have to make stronger 
assumptions and employ more complicated estimation methods than we do 
here. The same is true for \cite{freyberger18}, who considers
a non-separable model with interactive fixed effects.
Alternative non-linear extensions of factor models are discussed, for example, in \cite{cunha2010estimating} and \cite{gunsilius2019independent}.

The rest of the paper is organized as follows. Section~\ref{sec:estimators} introduces
our suggested estimators and inference methods.
Section~\ref{sect:consistencyBAI} and Section~\ref{sec:GroupedFE}
provide asymptotic results
for the LS estimator of  \cite{Bai2009}
and for our new two-way group fixed-effect estimator, respectively.
Section~\ref{sec:Implementation} discusses the practical
implementation.
Monte Carlo simulations are presented in Section~\ref{sect:simulations},
and an empirical application is worked out in Section~\ref{sec:empirical}.

\section{Estimation approaches}
\label{sec:estimators}

In this section, we introduce the two estimation approaches
that are afterwards analyzed and used in the rest of the paper.

\subsection{Least-squares interactive fixed effect estimator}

Following \cite{Bai2009} we consider
\begin{align}\label{est:bai}
   \left( \widehat \beta_{\rm LS} , \, \widehat \lambda, \, \widehat f \right)   
   &=
    \argmin_{(\beta,\lambda,f) \in \mathbb{R}^{K+N \times R+T \times R}}  \,
    \sum_{i=1}^{N}\sum_{t=1}^{T}\left(Y_{it} - X_{it}^\prime \, \beta - \sum_{r = 1}^{R}{\lambda}_{ir} \, {f}_{tr}\right)^2 .
\end{align}
This estimator was introduced for the exact factor model in equation \eqref{model:truncated}, and \cite{Bai2009} shows that it is $\sqrt{NT}$-consistent and asymptotically normally distributed for $N,T \rightarrow \infty$ when the true number of factors is fixed
and known. \cite{MoonWeidner2015} extend this result to the case where the true number
of factors is chosen too large in the estimation. 

To make the estimates $\widehat \lambda$ and $\widehat f$ in \eqref{est:bai} unique, we choose the usual normalization 
$T^{-1} \widehat f' \widehat f = \mathbb{I}_R$,
and $\widehat \lambda' \widehat \lambda$ to be a diagonal
matrix. In addition, it is convenient to introduce the notation
$X \cdot \beta$ for the $N \times T$ matrix with elements  $X_{it}' \beta$.

As explained above, the model \eqref{model:general} that we consider in this paper
can be rewritten as the factor model in \eqref{model:rewritten} with an infinite
number of factors in the true data generating process. This suggests that
the least-squares estimators in \eqref{est:bai} can still be consistent 
as long as the number of factors $R=R_{NT}$ used in the estimation is allowed
to grow to infinity jointly with $N$ and $T$.
Estimation of $\left( \widehat \beta_{\rm LS} , \, \widehat \lambda, \, \widehat f \right)$ is done using an iterative scheme. That is, we start by initialising $\widehat \beta_{\rm LS}$, and then iterate between estimating the principal components of $Y - X\cdot\widehat \beta_{\rm LS}$ to obtain $\left(\, \widehat \lambda, \, \widehat f \right)$  and least squares of $Y = X\cdot\beta + \widehat \lambda\widehat f^\prime + e$ to obtain $\widehat \beta_{\rm LS}$. The convergence metric we use is the sum of squares in \eqref{est:bai}. However, this iteration scheme can converge to a local minimum, and it is therefore important to repeat the procedure with multiple
starting values of $\beta$. For more details on the numerical computation of 
the estimator in \eqref{est:bai} we refer to \cite{Bai2009}
and \cite{MoonWeidner2015}.

This least-squares estimator of \cite{Bai2009} is very well-established in the panel regression
literature. It is used regularly both in empirical and in 
methodological papers, e.g.\ \cite{su2013testing}, \cite{KimOka2014},
\cite{lu2016shrinkage}, \cite{gobillon2016regional}, \cite{totty2017effect}, \cite{su2017time}, \cite{MoonWeidner2017}, \cite{giglio2021asset},
to name just a few.

\subsection{Group fixed effects estimator}
\label{subsec:GroupFEidea}

Here, we introduce two-way grouped fixed effects estimator, which 
discretizes the unobserved heterogeneity that is parameterized by $\alpha_i$ and $\gamma_t$ in the 
spirit of \cite{bonhomme2017discretizing}. We first describe the main 
idea of this estimator  before explaining its practical implementation in more details.

\subsubsection{Main idea}

We partition the set $\{1,\ldots,N\}$ of cross-sectional units
into $G=G_{NT}$ groups such that individuals in the same group have similar
values of $\alpha_i$. Let $g_i \in \{1,\ldots,G\}$ denote the group membership 
of individual $i$. Analogously, we partition the set $\{1,\ldots,T\}$ of time periods
into $C=C_{NT}$ groups such that time periods in the same group have similar
values of $\gamma_t$. Let $c_t \in \{1,\ldots,C\}$ denote the group membership 
of time period $t$. Details on how we construct those partitionings in practice are described below. Notice that within each
group the values of the $\alpha_i$ and $\gamma_t$, respectively,
need not be the same, but in the asymptotic theory in Section~\ref{sec:GroupedFE} the 
differences of those fixed effects within each group
are asymptotically negligible.

Once we have obtained those groups, then we estimate $\beta$ by applying
pooled OLS to the linear fixed-effect model
\begin{align}
   Y_{it} &=   X_{it}^\prime \, \beta + \delta_{i,c_t} +  \nu_{t,g_i}  + \epsilon_{it} ,
   \label{TwoWayFEbasic}
\end{align}
where $\delta_{i,c_t} \in \mathbb{R}$ and  $\nu_{t,g_i} \in \mathbb{R}$ are
nuisance parameters that are jointly estimated with $\beta$, that is, the 
basic two-way grouped fixed effect estimator
for $\beta$ can be written as
\begin{align}\label{est:clusterFE}
    \widehat \beta_{\rm G} &=
    \argmin_{\beta \in \mathbb{R}^K} \,
    \min_{\delta \in \mathbb{R}^{N \times C}} \,
    \min_{\nu \in \mathbb{R}^{T \times G}} \,
    \sum_{i=1}^{N}\sum_{t=1}^{T}\left(Y_{it} -   X_{it}^\prime \, \beta - \delta_{i,c_t} -  \nu_{t,g_i} \right)^2.
\end{align}
Notice that within each pair of groups for $i$ and $t$, that is, for
fixed values of $c_t$ and $g_i$, the model in \eqref{TwoWayFEbasic} is simply 
a standard additive two-way fixed effect model 
$ Y_{it} =   X_{it}^\prime \, \beta + \delta_{i} +  \nu_{t}  + \epsilon_{it}$.
However, as the group membership changes we allow the parameters $\delta_i$
and $\nu_t$ to change arbitrarily, as indicated by the additional
subscripts $c_t$ and $g_i$ in \eqref{TwoWayFEbasic}. We could have written
$\delta_{i,g_i,c_t} +  \nu_{t,g_i,c_t}$   to indicate 
explicitly that both the individual and time effect are allowed to change across groups,
but the notation in \eqref{TwoWayFEbasic} of course already allows for that generality.
The parameters
$\delta$ therefore form an $N \times C$ matrix, while the parameters $\nu$
form a $T \times G$ matrix.

In the introduction, we explained how the LS-estimator with interactive effects can be justified for model \eqref{model:general} by a truncation
of the functional singular value expansion in \eqref{functionSVD}. In other words,
a particular approximation of the function $\fct(\alpha_i,\gamma_t)$ naturally leads to
the estimator in~\eqref{est:bai}.

The grouped fixed effect estimator in \eqref{est:clusterFE}
can be justified analogously by a different approximation of the function $\fct(\alpha_i,\gamma_t)$. Under appropriate regularity conditions,
by a joint Taylor expansion in $\alpha_i$
and $\gamma_t$ around the corresponding 
group means $\overline \alpha_{g_i} = \frac{\sum_{j=1}^n \mathbbm{1}\{ g_i=g_j \} \, \alpha_j} {\sum_{j=1}^n \mathbbm{1}\{ g_i=g_j \}}$
and 
$\overline \gamma_{c_t} = \frac{\sum_{s=1}^T \mathbbm{1}\{ c_t=c_s \} \, \gamma_s} {\sum_{s=1}^T \mathbbm{1}\{ c_t=c_s \}}$,
we find that
\begin{align}\label{eq:taylorApprox}
   \fct(\alpha_i,\gamma_t)
 &= \delta_{i,c_t} + \nu_{t,g_i} + O\left(\|\alpha_i -\overline \alpha_{g_i} \|^2 + \|\gamma_t - \overline \gamma_{c_t} \|^2 \right) ,
\end{align}
where for vectors $\|\cdot\|$ denotes the Euclidean norm, and
\begin{align*}
    \delta_{i,c_t} &:= \fct(\overline \alpha_{g_i},\overline \gamma_{c_t}) +  
   \frac{\partial \fct(\overline \alpha_{g_i},\overline \gamma_{c_t})} {\partial \alpha'_i} 
   (\alpha_i - \overline \alpha_{g_i}) ,
   &
   \nu_{t,g_i} &:=  \frac{\partial \fct(\overline \alpha_{g_i},\overline \gamma_{c_t})} {\partial \gamma'_t} 
   (\gamma_t - \overline \gamma_{c_t}) .
\end{align*}
This shows that the leading order dependence of $\fct(\alpha_i,\gamma_t) $
on $\alpha_i$ and $\gamma_t$ can be described by the additive specification
$\delta_{i,c_t} + \nu_{t,g_i}$ used in \eqref{TwoWayFEbasic}. 
Since this two-way grouped fixed effect ignores the terms  
$O\left(\|\alpha_i -\overline \alpha_{g_i} \|^2 + \|\gamma_t - \overline \gamma_{c_t} \|^2 \right)$
entirely, it is of course crucial to construct the groups such that 
$\alpha_i -\overline \alpha_{g_i} $ 
and $\gamma_t - \overline \gamma_{c_t}$ are small. The clustering algorithm that
we use to achieve that is described in Subsection~\ref{Sect:clusterFE} below.

Notice that a naive application of   \cite{bonhomme2017discretizing}
to our two-way fixed effect model would {\it not} result in our estimating equation
\eqref{TwoWayFEbasic} but  in  $Y_{it}  =   X_{it}^\prime \, \beta + \chi_{g_i,c_t}  + \epsilon_{it}$, where $\chi_{g,c}$ is a fixed effect specific to each
pair of groups $(g,c) \in \{1,\ldots,G\} \times \{1,\ldots,C\}$. 
The analog of equation \eqref{eq:taylorApprox}
for that alternative approach reads
\begin{align*}
   \fct(\alpha_i,\gamma_t)  
   &=  \underbrace{\fct(\overline \alpha_{g_i},\overline \gamma_{c_t})  
   }_{=\chi_{g_i,c_t}} 
      + O\left(\|\alpha_i - \overline \alpha_{g_i}\| + \|\gamma_t -  \overline \gamma_{c_t} \| \right),     
\end{align*}
that is, the approximation error would be of linear order in the
discrepancies $\alpha_i - \overline \alpha_{g_i}$
and $\gamma_t - \overline \gamma_{c_t}$ within groups. By contrast, for our estimating equation \eqref{TwoWayFEbasic}
the resulting approximation error in \eqref{eq:taylorApprox} is of quadratic order, which explains why we prefer that approach.

Finally, notice that if our original model would only contain 
individual specific fixed effects $\alpha_i$, that is,
$Y_{it} = X_{it}^\prime \, \beta + \fct(\alpha_i) + \varepsilon_{it}$,
then the analog of \eqref{TwoWayFEbasic} is the standard additive fixed effect
model $Y_{it} = X_{it}^\prime \, \beta + \delta_i + \varepsilon_{it}$, which
requires no grouping at all, and also entails no approximation error
since we can set $\delta_i = \fct(\alpha_i)$. The way in which we generalize the
grouping ideas in \cite{bonhomme2017discretizing} is therefore quite specific
to the two-way fixed effect model in \eqref{model:general}.

\subsubsection{Hierarchical clustering algorithm}
\label{Sect:clusterFE}

To make the  two-way grouped fixed effect estimator   in \eqref{est:clusterFE}
operational we employ the following three-step algorithm:
\begin{enumerate}[A.]
    \item Obtain the factor loading and factor estimates 
    $\widehat \lambda$ and $\widehat f$ 
    of the interactive fixed effect LS estimator in \eqref{est:bai}
    for a relatively large number of factors $R$. Only keep the leading few 
    $R^*$ factor loading and factor estimates and denote those
    by $\widehat \lambda^* = (\widehat \lambda_{ir} \, : \, i=1,\ldots,N, \;
    r= 1,\ldots,R^*)$ and $\widehat f^* = (\widehat f_{tr} \, : \, t=1,\ldots,T, \;
    r= 1,\ldots,R^*)$.

\begin{table}[tb]
\begin{algorithm}[H]
\caption*{\bf Algorithm}
\begin{algorithmic}[1]
  \scriptsize
  \State Input $\widehat {\lambda}^*_i \in \mathbb{R}^{R^*}$ for all $i = 1,\dots N$. 
   Calculate all pairwise Euclidean distances $A_{ij} = \norm{\widehat {\lambda}^*_i - \widehat {\lambda}^*_j}$, for $i \neq j$, and set $A_{ii} = \infty$.
 Initialize ${\cal P} = \{ \{1\},\{2\}, \ldots,\{N\} \}$
 as a partition of $\{1,\ldots,N\}$.
  
    \If{$\exists \, {\cal C}_* \in {\cal P}$ with  $|{\cal C}_*|=4$} for that ${\cal C}_*$
        \State\label{algo:combHard} 
         Find the solution to $$\min_{\left\{ i,j,l,m  \, : \, {\cal C}_*=\{i,j,l,m\} \right\}}{A_{ij} + A_{lm}},$$
         \phantom{aa\,} 
         and split ${\cal C}_*$ into $\{i,j\}$ and $\{l,m\}$, updating 
         the partition ${\cal P}$.
    \ElsIf{$\exists  \, {\cal C} \in {\cal P}$ with  $|{\cal C}|=1$}
        \State Find the solution to 
        $$ \min_{\left\{ i \in \bigcup_{\{{\cal C} \in {\cal P}: |{\cal C}|=1\}} \right\}}
         \min_{\left\{j \in \bigcup_{\{{\cal C} \in {\cal P}: |{\cal C}|\leq 3\}}\right\}} A_{ij},$$
         \phantom{aa\,} and merge the clusters containing $i$ and $j$ into a single cluster, updating the partition ${\cal P}$.
  \EndIf
  \State Repeat 2-6 until $\{ |{\cal C}| \, : \, {\cal C} \in {\cal P}\} \subset \{2,3\}$.
\end{algorithmic}
\end{algorithm}
\caption{\label{algo:hier.cluster} Hierarchical clustering with minimum single linkage.}
\end{table}

    \item Use the $\widehat \lambda^*_1,\ldots,\widehat \lambda^*_N$ as
    inputs into the clustering algorithm in Table~\ref{algo:hier.cluster}
    to partition the set of individuals $\{1,\ldots,N\}$. This algorithm  
    returns the number $G$ of chosen groups and the group membership $g_i \in \{1,\ldots,G\}$ of each individual. Analogously, we use the inputs
    $\widehat f^*_1,\ldots,\widehat f^*_T$ into the same algorithm to
    partition $\{1,\ldots,T\}$, resulting in the 
    number of groups $C$ and the group membership $c_t \in \{1,\ldots,C\}$
    for each time period. Notice that the words partition, cluster, and group are used interchangeably in this paper.
    
    \item Calculate the two-way grouped fixed effect estimator $\widehat \beta_{\rm G}$
    via pooled OLS according to equation \eqref{est:clusterFE}.

\end{enumerate}

It is constructive to briefly describe our algorithm from Table~\ref{algo:hier.cluster} in words before we discuss features of this whole procedure. Step 1 defines the proxy variable to cluster on ($\widehat\lambda_i^*$ in this instance) and sets the distance metric we wish to use, Euclidean distance, which could easily be changed to another norm or metric. Then, we initialise each individual into their own cluster. Steps 2 and 3 then splits any groups of four into two groups of two, since we want groups of no larger than three in our final output.\footnote{%
We avoid singleton groups, because for those groups the within transformation
removes all information of the data. The restriction to groups
of at most size three is somewhat arbitrary, but we want to 
maintain small group sizes to guarantee that the differences 
in the fixed effects within each group are small, and there is no incidental parameter problem for the linear
fixed effect model in \eqref{TwoWayFEbasic}.
}
The optimisation in Step 3 looks at all combinations of two by two splits within this group of four and takes the smallest sum of distances. This type of optimisation is only suitable for very small groups of individuals because it is a combinatorially hard problem. 

Steps 4 and 5 then finds the solitary individual with the smallest distance to any other existing cluster and merges it to that cluster. Combined with Steps 2 and 3 we create an iteration that merges single clusters one at a time to groups of one, two or three, then splits any groups of four as and when they occur. This means Step 2 can only ever return one group of four. Doing this iteration one at a time is important so that we may split these groups of four immediately and have a larger choice set in Step 5 for each unmatched individual. Also, splitting groups of four into two by two groups rather than groups of one and three avoids infinite iterations. The repetition of Steps 2-5 is guaranteed to converge, and delivers a partition
of $\{1,\ldots,N\}$ into groups of size two or three.

Now to discuss the procedure as a whole. The choice of $R$ in step A here is not too important since we only need this to generate proxy variables for clustering and otherwise dispose of $\beta$ estimates from this initial LS step. The important hyperparameter is the number of proxies per observation, $R^*$, which we choose equal to two to five. We discuss the theoretical properties of the hyperparameter in Section \ref{sec:GroupedFEconsistency} but here outline a heuristic approach to this choice. Choosing $R^*$ to be more than one is important to capture cases when $\alpha_i$ and $\gamma_t$ have higher dimension or when the function $\fct(.,.)$ admits eigenfunctions that are not individually injective maps from $\alpha_i$ or $\gamma_t$. The aim is that a linear combination of non-injective maps provides a better mapping to the closeness of the primitives $\alpha_i$ and $\gamma_t$. An archetypal example of this is discussed in \cite{griebel2014approximation} where they show that the first few eigenfunctions of the exponential kernel are individually clearly not injective maps. 

It is also important to not use too many proxies so as to avoid clustering on noise. This can make for poor matches that result in large deviations between $\alpha_i$ and $\alpha_j$, respectively $\gamma_t$ and $\gamma_s$, that show up in the leading $O\left(\|\alpha_i - \alpha_j \|^2\right)$ and $O\left(\|\gamma_t - \gamma_s \|^2 \right)$ remainder  terms in \eqref{eq:taylorApprox}. Maintaining closeness in these primitives when clustering is key to any argument using Taylor's theorem, however, optimising this proxy hyperparameter is still rough and does require further development. We defer discussion about the presence of noise in factors with relation to the LS estimator to Section \ref{sect:consistencyBAI}.

There are, of course, other choices for proxies such as the cross-sectional moments employed in \cite{bonhomme2017discretizing}. However, as displayed in \eqref{functionSVD} and formulated in \cite{griebel2014approximation}, using the eigenfunctions from the singular value decomposition are a more natural choice since these are direct functions of the primitives $\alpha_i$ and $\gamma_t$ and should in theory lead to closer proximity between these. Since we require cross-sectional and time-dependent clusters for our method, these eigenfunctions also provide a convenient means to find these. 
If one truly believes that other proxy variables have more precise injectivity with these primitives then they could always make those the the input to Step 1 in our clustering algorithm. 

Another divergence from the existing literature is the use of clusters of size two or three, rather than letting these cluster sizes grow with sample size. Our motivation for using these small cluster sizes comes directly from the within-group $\beta$ estimation, i.e. that we do not need consistent estimates of $\delta$ or $\nu$ since these are treated as nuisance parameters that are simply differenced out. Hence, for our purposes, it is more useful to have small groups that are very similar rather than to have large groups that have better central tendency estimates. This very conveniently removes one choice for the analyst, namely the setting of group sizes $G$ or $C$.

This procedure is also a departure from the $k$-means approach taken in \cite{bonhomme2017discretizing}. 
For example, $k$-means with $k \approx N/2$ or $k \approx N/3$ only requires group sizes to be 2 or 3 on average.  This allows for 
a large heterogeneity in group sizes, which we avoid with our hierarchical approach.
Considering the distance metric in our algorithm is interchangeable, we expect our method to produce similar allocations to a $k$-means approach that manually limits cluster sizes to 2 or 3. 

Other cluster methods also exist. For example, in the presence of heterogeneous coefficients~$\beta_i$, \cite{su2016identifying} and \cite{su2019sieve} propose clustering on $\beta_i$. 
The procedure proposed there may suggest useful ways to incorporate heterogeneity in the slope coefficients in this setting, or indeed provide good cluster proxies for the unobserved heterogeneity term. 
It should be noted, however, that in those settings there exists a true group structure, which departs from our approach that considers groups as useful discretisations of the underlying parameter space.

\subsubsection{Split-sample version of the estimator}
\label{sec:DefinedBetaGS}

As explained above,
we estimate the group memberships $g_i$ and $c_t$ that enter into 
the estimator for $\beta$ in \eqref{est:clusterFE} via a clustering method applied
to $\widehat \lambda^*$ and $\widehat f^*$. However, clustering in this way creates dependence across $i$ and $t$ through $\widehat \lambda^*$ and $\widehat f^*$. This dependence 
creates technical difficulties when
establishing asymptotic convergence results. 
To mitigate this dependence we augment the clustering
estimator by a simple sample splitting method.
The resulting group fixed effect estimator with sample splitting 
is given by
\begin{align}\label{est:clusterFEsplit}
    \widehat \beta_{\rm GS} &=
    \argmin_{\beta \in \mathbb{R}^K} \,
    \min_{\delta} \,
    \min_{\nu} \,
    \sum_{i=1}^{N}\sum_{t=1}^{T}\left[Y_{it} -   X_{it}^\prime \, \beta -
    \sum_{s=1}^S
    \mathbbm{1}\left\{ (i,t) \in {\cal O}_{s}  \right\} 
    \left( \delta^{(s)}_{i,c^{(s)}_t} +  \nu^{(s)}_{t,g^{(s)}_i}
    \right) \right]^2 ,
\end{align}
where $S$ is the number of partitions,
and the sets ${\cal O}_{s}$, $s=1,\ldots,S$, are the partitions of the sample space
$\{1,\ldots,N\} \times \{1,\ldots,T\}$, that is,
the observation $(i,t)$ is a member of the $s$'th partition 
if and only if $(i,t) \in {\cal O}_{s}$.
Compared to the original group fixed effect estimator
in \eqref{TwoWayFEbasic}, the group membership indicators $g^{(s)}_i$ and $c^{(s)}_t$ and the group fixed effect
$\delta^{(s)}_{i,c^{(s)}_t}$ and $\nu^{(s)}_{t,g^{(s)}_i}$
are all specific to the partition $s$.
For the purpose of this paper, we choose the number of 
partitions to be $S = 4$ and we split the sample space into four blocks as follows:
\begin{align}\label{eqn:partitions}
    \begin{split}
        \mathcal{O}_1 &=  \{1,\dots,\lfloor N/2\rfloor\} \times \{1,\dots,\lfloor T/2\rfloor\},  \\ 
        \mathcal{O}_2 &=  \{1,\dots,\lfloor N/2\rfloor\} \times \{\lfloor T/2\rfloor + 1,\dots,T\},  \\ 
        \mathcal{O}_3 &= \{\lfloor N/2\rfloor + 1,\dots, N\} \times \{1,\dots,\lfloor T/2\rfloor \},  \\
        \mathcal{O}_4 &=  \{\lfloor N/2\rfloor + 1,\dots,N\} \times \{\lfloor T/2\rfloor + 1,\dots,T\} ,
    \end{split}
\end{align}
where $\lfloor \cdot \rfloor$ is the floor function.

We still need to explain how the group memberships $g^{(s)}_i$ and $c^{(s)}_t$ are obtained here.
The aim of the sample splitting is to avoid any stochastic dependence  between  $g^{(s)}_i$ and $c^{(s)}_t$ and the idiosyncratic noise $\varepsilon_{it}$. For each partition
$s=1,\ldots,S$, we therefore construct the 
group memberships $g^{(s)}_i$ and $c^{(s)}_t$ without
using outcomes $Y_{it}$ for observations $(i,t)$ of that partition $\mathcal{O}_s$.
For that purpose, we define the sets
\begin{align}\label{eqn:partitionsProxies}
    \begin{split}
        \mathcal{O}^*_1 &= \{1,\dots,N\} \times  \{1,\dots,\lfloor T/2\rfloor\}, \\ 
        \mathcal{O}^*_2 &= \{1,\dots,N\} \times  \{\lfloor T/2\rfloor + 1,\dots,T\}, \\ 
        \mathcal{O}^*_3 &= \{1,\dots,\lfloor N/2\rfloor\} \times  \{1,\dots,T\}, \\
        \mathcal{O}^*_4 &=  \{\lfloor N/2\rfloor  + 1,\dots,N\} \times  \{1,\dots,T\} ,
    \end{split}
\end{align}
and for $\tilde s=1,\ldots,4$, we define the corresponding least-squares factor and loading
estimates
\begin{align}\label{est:baiPartitions}
   \left( \, \widehat \lambda^{(\tilde s)}, \, \widehat f^{(\tilde s)} \right)   
   &=
    \argmin_{(\lambda,f) \in \mathbb{R}^{N^*_{\tilde s} \times R+T^*_{\tilde s} \times R}}  \, \min_{\beta\in\mathbb{R}^K}
    \sum_{(i,t)\in\mathcal{O}^*_{\tilde s}}\left(Y_{it} - X_{it}^\prime \, \beta - \sum_{r = 1}^{R}{\lambda}_{ir} \, {f}_{tr}\right)^2 ,
\end{align}
which is simply the LS estimator in \eqref{est:bai}
applied only to the $N^*_{\tilde s} \times T^*_{\tilde s}$ subpanel
of observations $(i,t)\in\mathcal{O}^*_{\tilde s}$,
and we also impose the same normalization on the factors
and loadings explained after \eqref{est:bai}.\footnote{
Notice that factor model proxies can only be used to compare observations from the same factor estimation sample space. This is because factors are only identified up to rotations, where these rotations may differ across estimation samples. 
}
Now, for the original partition $\mathcal{O}_s$, $s=1,\ldots 4$,
we construct the group membership $g^{(s)}_i$ of unit $i$
by applying the clustering algorithm in Table~\ref{algo:hier.cluster} to the
loading estimates $\widehat \lambda^{(\tilde s)}_i$
obtained from the subpanel $\mathcal{O}_{\tilde s}$
with $\tilde s=\tilde s(s)$ given by
\begin{align*}
   \tilde s = \left\{ \begin{array}{ll}
        2 & \text{for} \; s=1, \\
        1 & \text{for} \; s=2, \\
        2 & \text{for} \; s=3, \\
        1 & \text{for} \; s=4.
   \end{array} \right.
\end{align*}
Analogously, for the partition $\mathcal{O}_s$, $s=1,\ldots 4$,
we construct the group membership $c^{(s)}_t$ of time period $t$
by applying the clustering algorithm in Table~\ref{algo:hier.cluster} to the
factor estimates $\widehat f^{(\tilde s)}_t$
obtained from the subpanel $\mathcal{O}_{\tilde s}$
with $\tilde s=\tilde s(s)$ given by
\begin{align*}
   \tilde s = \left\{ \begin{array}{ll}
        4 & \text{for} \; s=1, \\
        4 & \text{for} \; s=2, \\
        3 & \text{for} \; s=3, \\
        3 & \text{for} \; s=4.
   \end{array} \right.
\end{align*}

\begin{figure}
    {\centering
    \includegraphics[scale = 0.45]{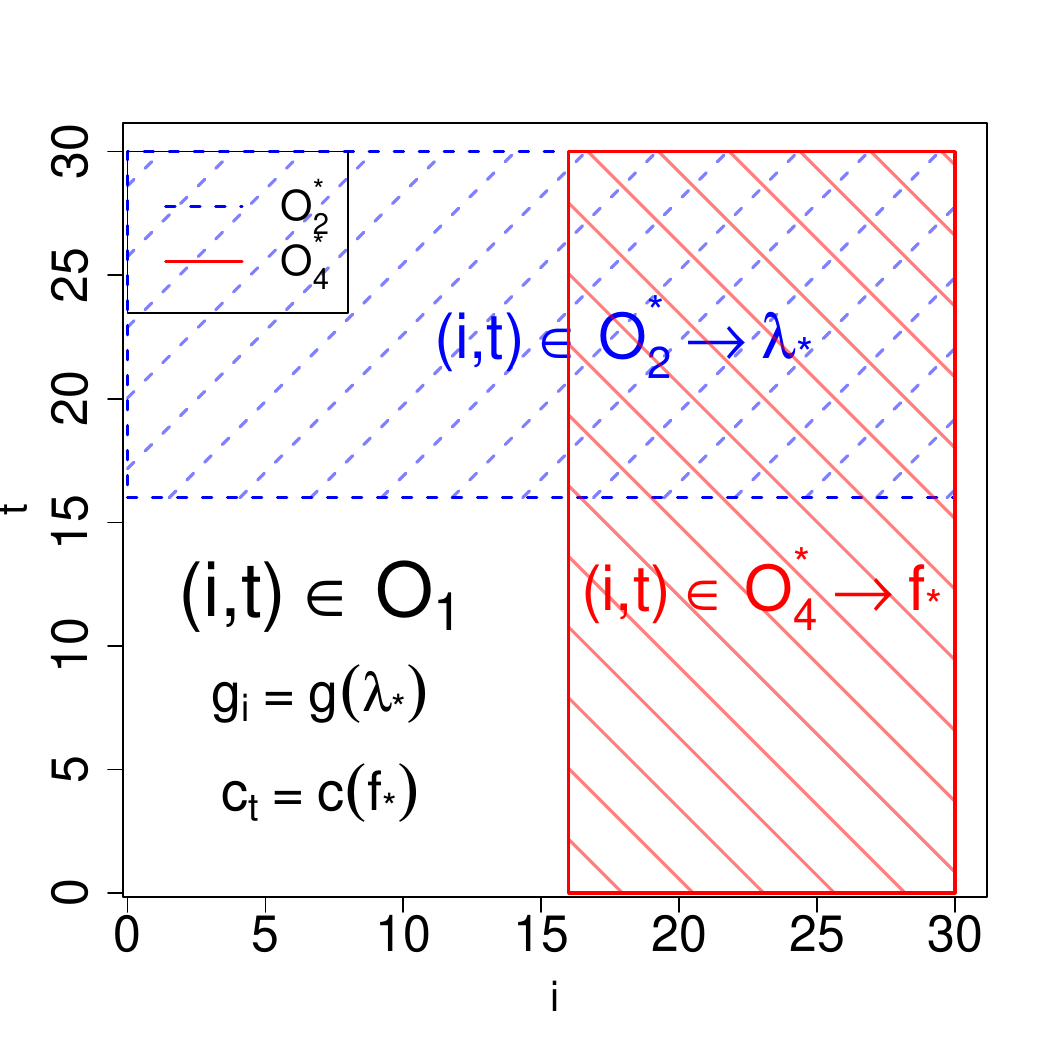}
    \caption{Sample split for partition 1}
    \label{fig:SplitSample}
    }

\end{figure}

Figure \ref{fig:SplitSample} details an example of this sample splitting technique for clustering within partition $\mathcal{O}_1$. 
Here we see clearly how the partitions for proxy estimation $\mathcal{O}^*_2$ and $\mathcal{O}^*_4$ do not overlap with the partition we are grouping within, $\mathcal{O}_1$. 
This guarantees that  we do not introduce any dependence between the group functions and the noise term by making sure grouping within each partition is not a function of the independent noise term, $\varepsilon_{it}$, from observations within that partition. 
This becomes important in our derivations in Section~\ref{sec:GroupFEsplitnormal}, where
we require that the process $X_{it} \, \varepsilon_{it}$ remains zero mean and independently distributed after group means are projected out.

With these cluster assignments it then becomes straightforward to estimate \eqref{est:clusterFEsplit} by first taking within-cluster mean-differences for each partition and then simply apply pooled OLS on the transformed variables.

Notice that by allowing the partitioning in \eqref{eqn:partitionsProxies} used to estimate proxy variables to extend over the whole sample of either $N$ or $T$, we get better estimates than just using the original partition 
\eqref{eqn:partitions}.
As discussed earlier, it is crucial to avoid poor initial estimates of proxy variables to better approximate the residual terms in the Taylor expansion in expression \eqref{eq:taylorApprox}.

\section{Asymptotic results for the least squares estimator}\label{sect:consistencyBAI}

Here,
we derive convergence rate results for the 
least-squares estimator \eqref{est:bai} 
for a data generating process given by \eqref{model:general}.
Thus, we generalize the consistency results in \cite{Bai2009}
and \cite{MoonWeidner2015} to the case where the underlying 
panel regression model does not satisfy the factor model 
in \eqref{model:truncated}. However, as explained in the introduction,
the factor model in \eqref{model:truncated} can be viewed as an approximation
of \eqref{model:general}, and this approximation idea can be formalized asymptotically, as long as we allow the number of factors $R=R_{NT}$
used in the least-squares estimator \eqref{est:bai} to grow with $N$
and $T$.

\subsection{Consistency and convergence rate}\label{sect:bai.consistency}
From now on, we denote the true parameter $\beta$ that generates the data by $\beta^0$.
We rewrite model \eqref{model:general} as
\begin{align}\label{eq:consistency.model:general}
    Y_{it} &= X_{it}^\prime \, \beta^0 + \Gamma_{it} + \varepsilon_{it}, 
\end{align}
where both $\Gamma_{it}$ and $\varepsilon_{it}$ are unobserved.
Our main convergence rate results in Theorem~\ref{thm:bai.con}
actually hold
for any $N\times T$ matrix $\Gamma=(\Gamma_{it})$ that satisfies Assumption~\ref{ass:SVD} below, but ultimately we are 
of course interested in the case $\Gamma_{it} = \fct(\alpha_i,\gamma_t)$.
Arbitrary dependence between $X_{it}$ and $\Gamma_{it}$ is allowed for, so there is a potential endogeneity problem. 

Remember that the components of the $K$-vector $X_{it}$ are denoted by $X_{it,k}$,
$k=1,\ldots,K$.
Let $X_k=(X_{it,k})$ and $\varepsilon=(\varepsilon_{it})$ be
$N \times T$ matrices. 
For a matrix $A$ we denote $r$'th largest singular value
by $\sigma_r(A)$, that is, $\sigma^2_r(A)$ is equal to the 
$r$'th largest eigenvalue of $AA'$.
Furthermore, for
 matrices we denote the spectral
norm by $\| \cdot \|$, and for vectors the norm $\| \cdot \|$ denotes
the Euclidean norm. We write wpa1 for ``with probability approaching
one''.
We impose the following assumptions.

\begin{assumption}[\bf Bounded norms of $X_k$ and $\varepsilon$]~
\label{ass:SN}
\begin{itemize}
  \item [(i)] $\displaystyle \frac 1 {NT} \sum_{i=1}^N \sum_{t=1}^T (X_{it,k})^2 = { O}_P(1)$, \quad \text{for} $k=1,\ldots,K$.
  \item [(ii)] $\displaystyle \| \varepsilon \| = { O}_P\left(\sqrt{\max\{N,T\}} \right)$. 
\end{itemize}
\end{assumption}

\begin{assumption}[\bf Weak Exogeneity of $X_k$]~
  \label{ass:EX}
  $\displaystyle \sum_{i=1}^N \sum_{t=1}^T X_{it,k} \varepsilon_{it}  = {O}_P(\sqrt{NT})$,
              \hfill   \text{for} $k=1,\ldots,K$.            
\end{assumption}

\begin{assumption}[\bf Non-Collinearity of $X_k$]~
  \label{ass:NC}
  Consider linear combinations
  $\delta \cdot X :=\sum_{k=1}^K \delta_k X_k$ of the regressors $X_k$ with
  vectors $\delta \in \mathbb{R}^K$ such that $\|\delta\|=1$.
  Assume that there exists a constant $b>0$ such that
  \begin{align*}
     \min_{\{\delta \in \mathbb{R}^{K}, \, \|\delta\|=1\}}   \,
     \sum_{r=2R_{NT}+1}^{\min(N,T)} \,
      \sigma^2_{r} \left[ \frac{(\delta \cdot X)} {\sqrt{NT}}\right] \;  &\geq b \; , \qquad
     \text{wpa1.}
  \end{align*}
\end{assumption}

\begin{assumption}[\bf Singular value decay]~
  \label{ass:SVD}
   There exists a constant  $\rho>3/2$ such that
   $$\frac 1 {NT} \sum_{r=R_{NT}+1}^{\min(N,T)}
   \sigma^2_r(\Gamma) = O_P \left( R_{NT}^{1-2\rho} \right). $$
\end{assumption}

\medskip

Here, $R=R_{NT}$ is the number of factors that is chosen 
in the computation of the 
least-squares estimator $ \widehat \beta_{\, \rm LS}$ in \eqref{est:bai}. 
We require  $R_{NT} \rightarrow \infty$ as $N,T \rightarrow \infty$
to obtain consistency of  $\widehat \beta_{\, \rm LS}$.

Lemma~\ref{lemma:SG} below justifies Assumption~\ref{ass:SVD} for our main case of interest $\Gamma_{it} = \fct(\alpha_i,\gamma_t)$, and we therefore postpone
the discussion of that assumption until we discuss that lemma.
 Assumptions~\ref{ass:SN}-\ref{ass:NC} are very similar
to the assumptions used in \cite{Bai2009}
and \cite{MoonWeidner2015}
 to show consistency of $ \widehat \beta_{\, \rm LS}$,\footnote{Compared to the assumptions imposed in the
consistency Theorem 4.1 of \cite{MoonWeidner2015},
the only two differences  are that we allow for $R_{NT}$
to grow asymptotically, and that Assumption~\ref{ass:SN}(i)
requires a bound on the Frobenius norm $\|X_k\|_F := \left( \sum_{i=1}^N \sum_{t=1}^T X_{it,k}^2 \right)^2$ instead of
a bound on the spectral norm $\|X_k\|$. Since
$\|X_k\| \leq \|X_k\|_F$, our assumption here is technically stronger, 
but in practice, one likely will justify any bound on $\|X_k\|$
using the inequality $\|X_k\| \leq \|X_k\|_F$ anyway.
}
and the following discussion of those assumptions will, accordingly,
be brief. 

Assumption~\ref{ass:SN}(i) follows from Markov's inequality
as long as the second moment of $X_{it,k}$ is uniformly bounded. Assumption~\ref{ass:SN}(ii)
follows, for example, from the inequality in \cite{Latala2006}
if $\varepsilon_{it}$ has mean zero, uniformly bounded fourth moment,
and is independent across $i$ and $t$. 
However, the assumption still holds if $\varepsilon_{it}$ is weakly
correlated across $i$ and over $t$, see \cite{MoonWeidner2015}.
Assumption~\ref{ass:EX} is satisfied as long as $X_{it} \varepsilon_{it}$
has zero mean, uniformly bounded second moment, and is weakly correlated across
$i$ and over $t$.

To understand Assumption~\ref{ass:NC}, notice first that
for $R_{NT}=0$ the expression $\sum_{r} \,
      \sigma^2_{r} \left[ \frac{(\delta \cdot X)} {\sqrt{NT}}\right]$
      in that assumption becomes
\begin{align*}
   \sum_{r= 1}^{\min(N,T)} \,
      \sigma^2_{r} \left[ \frac{(\delta \cdot X)} {\sqrt{NT}}\right]
      = \frac 1 {NT}  \sum_{i=1}^N \sum_{t=1}^T
      (\delta \cdot X)_{it}^2 .
\end{align*}
Thus, for $R_{NT}=0$, the assumption is just a standard non-collinearity
assumption on the regressors, which demands that every non-trivial 
linear combination $\delta \cdot X$ of the regressors has sufficient variation.
Next, for $R_{NT}>0$ we have
\begin{align*}
   \sum_{r= 2R_{NT} + 1}^{\min(N,T)} \,
      \sigma^2_{r} \left[ \frac{(\delta \cdot X)} {\sqrt{NT}}\right]
      = \frac 1 {NT}  \sum_{i=1}^N \sum_{t=1}^T
      (\delta \cdot X)_{it}^2 
      -  \sum_{r= 1}^{2R_{NT}} \,
      \sigma^2_{r} \left[ \frac{(\delta \cdot X)} {\sqrt{NT}}\right] ,
\end{align*}
that is, the assumption demands that the variation in the linear 
 combination $\delta \cdot X$ does not only come from the leading
 $2R_{NT}$ singular values of this linear combination. 
 
 Of course, if ${\rm rank}(\delta \cdot X) \leq 2 R_{NT}$, then
 for $r > 2R_{NT}$
all the singular values 
$ \sigma_{r} \left(  \delta \cdot X \right)$ 
are equal to zero
and the assumption is violated. 
Thus, a necessary condition for Assumption~\ref{ass:NC} is that
${\rm rank}(\delta \cdot X) > 2 R_{NT}$, that is,
any linear combination of the regressors needs to be a ``high-rank matrix''.
For example, a constant regressor
 $X_{it,1}=1$  violates this assumption (it constitutes a rank one matrix, 
which could be easily absorbed into the unobserved $\Gamma_{it}$),  
but if the regressors are drawn from a DGP with random variation across
both $i$ and $t$, then they typically have full rank.
Again, we refer to the existing papers on the least-squares estimator
with interactive fixed effects for further discussion
of this generalized non-collinearity condition on the regressors.

\medskip

\begin{theorem}[\bf Consistency of $\widehat \beta_{\, \rm LS}$]\label{thm:bai.con}
    Let Assumptions \ref{ass:SN} -- \ref{ass:SVD}
    hold, and furthermore assume that $R_{NT} = o(\min\{N,T\})$ as
    $N,T \rightarrow \infty$.
    Then we have
      \begin{align}
        \widehat \beta_{\, \rm LS} - \beta^0 =  
        O_P\Big(R_{NT}^{(3-2\rho)/2}\Big) +
        O_P\Big(R_{NT} \, (\min\{N,T\})^{-1/2}\Big). 
        \label{MainResultsLS}
    \end{align}
    Therefore, by choosing $R_{NT}\propto (\min\{N,T\})^{\frac{1}{2\rho-1}}$
    we obtain that
    \begin{align*}
        \widehat \beta_{\, \rm LS} - \beta^0 = {O}_P\Big(\min\{N,T\}^{\frac{3 - 2\rho}{2(2\rho-1)}}\Big).
    \end{align*} 
\end{theorem}

\medskip

Assumption~\ref{ass:SVD} demands $\rho>3/2$, and the
 first term on the right-hand side of \eqref{MainResultsLS} is therefore 
 decreasing in the number of factors $R_{NT}$ used for estimation.
 By contrast, the second term on the right-hand side of \eqref{MainResultsLS}
 is increasing in $R_{NT}$.
 The final part of the theorem simply gives the rate for $R_{NT}$ that optimally balances the trade-off between those two terms.
 This   is  
 analogous the  bias-variance trade-off for bandwidth selection 
 in non-parametric estimation.
Indeed,
the term $O_P\left(R_{NT}^{(3-2\rho)/2}\right)$ is 
due to the approximation error of the $N \times T$ matrix $\Gamma$ (which can have full rank)
by  only a finite number of factors (of rank only $R_{NT}$).
As expected, the approximation
error is small when choosing a more flexible model (large $R_{NT}$).

The second term on the right-hand side of \eqref{MainResultsLS} also
occurs when one considers a conventional interactive fixed effect model,
where the true matrix $\Gamma$ itself is assumed to have low rank
and the approximation error is therefore not present (for $R_{NT} \geq {\rm rank}(\Gamma)$).  
For that case, the first paper to derive the large $N$, $T$ asymptotic properties for
$ \widehat \beta_{\, \rm LS}$ was \cite{Bai2009}. He imposes
assumptions (in particular, $R_{NT} = {\rm rank}(\Gamma) = {\rm constant}$, and all  factors in $\Gamma$ are ``strong factors'') that are strong enough to derive the result 
$\widehat \beta_{\, \rm LS} - \beta^0 = O_P(1/\sqrt{NT})$ when $N$ and $T$
grow at the same rate.\footnote{For general sequences of $N, T \rightarrow \infty$ one finds $\widehat \beta_{\, \rm LS} - \beta^0 = O_P(1/N+1/T)$ under
 Bai's assumptions.
}
However, without such strong assumptions, the estimator $\widehat \beta_{\, \rm LS}$
may very well converge at a slower rate. 
For example, in Section~4.3 of \cite{MoonWeidner2015} a concrete
data generating process is given for which $ \widehat \beta_{\, \rm LS}$ only converges at the slower rate $ (\min\{N,T\})^{-1/2}$.\footnote{
In that example,
the unnecessarily estimated loadings $\widehat \lambda$ and factors 
 $\widehat f$ are correlated with  the regressors, and by controlling for such endogenous 
 $\widehat \lambda$
and $\widehat f$ one ends up reducing the convergence rate
of $ \widehat \beta_{\, \rm LS}$ from $\sqrt{NT}$
to $  (\min\{N,T\})^{1/2}$. 
}
The key difference between that
example and \cite{Bai2009} is that $R_{NT} > {\rm rank}(\Gamma)$, that is,
the number of factors in the estimation is larger than the true number of factors.
More generally, as soon as the ``strong factor'' assumption or the known number
 of factors assumption ($R_{NT} = {\rm rank}(\Gamma)$) are violated, there is no guarantee
 that $ \widehat \beta_{\, \rm LS}$ converges at the fast rate derived
 in \cite{Bai2009}. In the absence of those assumptions, Theorem~4.1 in
 \cite{MoonWeidner2015} shows that
$\widehat \beta_{\, \rm LS} - \beta^0 =  
        O_P\Big( (\min\{N,T\})^{-1/2}\Big)$ when $R_{NT} \geq {\rm rank}(\Gamma)$ is fixed.
The second term on the right-hand side of \eqref{MainResultsLS} exactly generalizes
that rate to the case where $R_{NT}$ is allowed to grow asymptotically.

In our setting, we cannot impose the ``strong factor'' 
or known number of factor assumptions in \cite{Bai2009}, because,
as explained in the introduction, the data generating process
$\Gamma_{it} = \fct(\alpha_i,\gamma_t)$ typically generates
an infinite sequence of factors of decreasing strength. Demanding
all those factors in equation \eqref{model:rewritten} to be strong
factors  makes no sense in our setting. Deriving a convergence rate for 
$ \widehat \beta_{\, \rm LS}$ faster than $ (\min\{N,T\})^{-1/2}$
in our model therefore
appears to very challenging, to say the least. This is of course,
the key motivation for why we also consider the two-way grouped fixed effect
estimator in this paper, see Section~\ref{sec:GroupedFE} below.

\begin{remark}
   If we change
   Assumption~\ref{ass:SVD} to   
   \begin{align}
   \sigma_r(\Gamma) \, \leq \, c \, \sqrt{NT} \, r^{-\rho} ,
       \label{Alernative:ass:SVD}
   \end{align}
    for all $r\in\{R_{NT}+1,\dots \min\{N,T\}\}$, wpa1,
   and some constant $c>0$,    
   then the result in equation \eqref{MainResultsLS} of Theorem~\ref{thm:bai.con} can be improved to
    \begin{align*}
        \widehat \beta_{\, \rm LS} - \beta^0 = 
        {O}_P\left( { R_{NT}^{\, 1-\rho} }\right) + 
        {O}_P\left({R_{NT}} \, {(\min\{N,T\})^{-1/2}}\right) ,
    \end{align*}
    and we can then obtain consistency of $\widehat \beta$ under the weaker
    condition $\rho>1$.
    Condition \eqref{Alernative:ass:SVD} implies  Assumption~\ref{ass:SVD},
    but not vice versa,
    because Assumption~\ref{ass:SVD} is a condition on the sum of 
    the squared singular values, not on each of the singular values separately.
    It turns out to be technically much easier to verify Assumption~\ref{ass:SVD}
    than to verify \eqref{Alernative:ass:SVD}
    for our main case of interest $\Gamma_{it} = \fct(\alpha_i,\gamma_t)$,\footnote{
    This is because not only the decay of  $\sigma_r(\Gamma)$
    as $r \rightarrow \infty$ needs to be controlled, but also the convergence
    rate of the expressions as $N,T \rightarrow \infty$.
    }
    as we do in Lemma~\ref{lemma:SG} below.
    This explains why we have chosen that formulation of the assumption
    and theorem
    in our baseline presentation. 
\end{remark}

Despite the technical subtleties explained in the preceding remark,
one should still interpret Assumption~\ref{ass:SVD} 
as imposing a particular decay rate for the singular values $\Gamma$,
as in display \eqref{Alernative:ass:SVD} of the remark. 
Thus, the leading few singular value
can have a magnitude of $\sqrt{NT}$, as would be the case under the 
``strong factor assumption'' in the usual interactive fixed effects model
of \cite{Bai2009}.
However, as $N$, $T$, $r$ all converge to infinity we require the $\sigma_r(\Gamma)$ to converge at the polynomial rate $r^{-\rho}$
in order to satisfy the summability condition in Assumption~\ref{ass:SVD}.

The results in this section so far have not made any use of the
structure $\Gamma_{it} = \fct(\alpha_i,\gamma_t)$.
Theorem~\ref{thm:bai.con} is applicable to 
any other data generating process for $\Gamma$ that satisfies Assumption~\ref{ass:SVD}.
A full-rank matrix $\Gamma$ satisfying that assumption could, for example,
also be generated by a dynamic factor model (see e.g.\ \citealt{forni2000generalized,forni2005generalized}, \citealt{stock2002macroeconomic}).\footnote{
One can generate an infinite number of ``static factors'', as in \eqref{model:rewritten}, via a dynamic factor model with a finite number of dynamic factors.}

In the following we now focus exclusively on the case $\Gamma_{it} = \fct(\alpha_i,\gamma_t)$. 
The following lemma provides conditions on the function 
$\fct(\cdot,\cdot)$ that guarantee that  Assumption~\ref{ass:SVD} is satisfied.

\begin{lemma}	\label{lemma:SG}
Assume 
 $\alpha_i \in \Omega_\alpha$
and $\gamma_t \in \Omega_\gamma$, and
that
    $\fct:\Omega_\alpha \times \Omega_\gamma \rightarrow \mathbb{R}$ is $p$ times continuously differentiable
    in both arguments,
    with uniformly bounded mixed-derivatives up to order $p$,
    and the domains $\Omega_\alpha\subset\mathbb{R}^{n_\alpha}$ and $\Omega_\gamma\subset\mathbb{R}^{n_\gamma}$ are smooth and bounded.
Then for $\Gamma_{it}=\fct(\alpha_i,\gamma_t)$  Assumption~\ref{ass:SVD}
is satisfied { for $R_{NT} \rightarrow\infty$} with $\rho =  \frac{p}{\min\left\lbrace n_\alpha , n_\gamma \right\rbrace}$.
\end{lemma}

Here, we measure the smoothness of the function $\fct(\cdot,\cdot)$ by $p$,
which is the number of times it is continuously differentiable. 
The decay rate $\rho$
of the singular values of $\Gamma$ then depends on this measure of smoothness
and the dimensions $n_\alpha$ and $n_\gamma$ of the arguments $\alpha_i$
and $\gamma_t$. 
The smoother the function $\fct(\cdot,\cdot)$, for fixed dimensions $n_\alpha$ and $n_\gamma$, the faster the eigenvalues of $\Gamma$ converge
to zero.

The proof of Lemma~\ref{lemma:SG} crucially relies on the functional singular
value decomposition in \eqref{functionSVD} and results on the decay rate 
of the corresponding singular values in \cite{griebel2014approximation}.
The only technical contribution of the proof  is then
to properly relate those known results on the functional singular value
to the matrix singular values of $\Gamma$.

Notice that Lemma~\ref{lemma:SG} requires no assumptions on the data
generating process of $\alpha_i$ and $\gamma_t$,
apart from boundedness of the domains
$\Omega_\alpha$ and $\Omega_\gamma$, which 
can always be achieved by a reparameterization.
Thus, those nuisance
parameters can be arbitrarily correlated with each other (across $i$ and over $t$)
and with the regressors $X_{it,k}$. This result is analogous
to the consistency Theorem~4.1 for  $\widehat \beta_{\, \rm LS}$ in \cite{MoonWeidner2015}, where also no
assumptions on the interactive fixed effects are imposed at all, apart from
 ${\rm rank}(\lambda f') \leq R$.

From Theorem~\ref{thm:bai.con} and Lemma~\ref{lemma:SG}  we have the following corollary.

\begin{corollary} \label{cor:LSconsistency}
    Let Assumptions \ref{ass:SN} -- \ref{ass:NC} and 
    the assumption on $\fct(.,.)$ in Lemma~\ref{lemma:SG} be satisfied
    with {$p > 3\min\left\lbrace n_\alpha , n_\gamma \right\rbrace/2$},
    and also let  $R_{NT} \rightarrow \infty$ such that ${R_{NT} / (\min\{N,T\})^{1/2} \rightarrow 0}$. Then we have
    \begin{align*}
        \widehat \beta_{\, \rm LS}  - \beta^0 = o_P(1) .
    \end{align*}
\end{corollary}
This is our final consistency result for the least-squares estimator
of \cite{Bai2009} in a data generating process given by \eqref{model:general}. The
convergence rate of the estimator was already discussed after 
Theorem~\ref{thm:bai.con} above, in particular, the difficulty 
in showing a convergence rate faster than $ (\min\{N,T\})^{1/2}$ 
in our setting.

\subsection{Further discussion}

Here, we want to present some further intuition on the formal results on $\widehat \beta_{\, \rm LS}$ presented above. The  discussion in this subsection is purely
heuristic and does not aim to provide any formal derivations.

Remember the functional singular
value decomposition in equation 
\eqref{functionSVD} of the introduction,
which we now write as
$\fct(\alpha_i,\gamma_t)
= \sum_{r=1}^{\infty} \lambda^0_{ir} \, f^0_{tr} $.
For the sake of the following discussion,
suppose that variation from $\fct(\alpha_i,\gamma_t)$ dominates the variation in $X_{it}^\prime \, \beta$ and $\varepsilon_{it}$ for the leading
$R_{NT}$ principal components
of the residuals
$Y_{it} - X_{it}^\prime \, \beta - \sum_{r = 1}^{R}{\lambda}_{ir} \, {f}_{tr}
= \sum_{r=1}^{\infty} \lambda^0_{ir} \, f^0_{tr} - X_{it}^\prime \, (\beta-\beta^0) + \varepsilon_{it}$.
In this ``best case scenario'',
the estimated factors
$\sum_{r = 1}^{R}{\lambda}_{ir} \, {f}_{tr}$
in the definition of $\widehat \beta_{\rm LS}$ in 
\eqref{est:bai}
will coincide with the leading
$R_{NT}$ components 
$\sum_{r = 1}^{R}{\lambda}^0_{ir} \, {f}^0_{tr}$
of $\fct(\alpha_i,\gamma_t)$, and
we then have
\begin{align*}
    \widehat \beta_{\rm LS} - \beta^0 = \zeta_{NT} + \xi_{NT} ,
\end{align*}    
where    
\begin{align*}
    \zeta_{NT}
    &=
    \left( \frac 1 {NT} \sum_{i=1}^N \sum_{t=1}^T X'_{it} \, X_{it} \right)^{-1}
      \frac 1 {NT} \sum_{i=1}^N \sum_{t=1}^T X'_{it} \, \varepsilon_{it}
 \\
     \xi_{NT}
    &=
    \left( \frac 1 {NT} \sum_{i=1}^N \sum_{t=1}^T X'_{it} \, X_{it} \right)^{-1}
     \frac 1 {NT} \sum_{i=1}^N \sum_{t=1}^T X'_{it} \, \sum_{r=R+1}^{\infty}\lambda^0_{ir}f^0_{tr} .
\end{align*}
Under standard regularity conditions we have  
$\sqrt{NT} \zeta_{NT} \Rightarrow {\cal N}(0,\Sigma)$,
and under the assumptions in the last subsection we have
$\xi_{NT} = {O}_P\Big(R_{NT}^{(3-2\rho)/2}\Big) $.
In this ``best-case scenario'' we can therefore have $R_{NT} \rightarrow \infty$ quick enough such that $\xi_{NT} = o_P(1/\sqrt{NT})$.

However, this is not a realistic scenario for
$R_{NT} \rightarrow \infty$, because 
as $R_{NT}$ grows, eventually the singular values
of $\varepsilon_{it}$ will dominate those of $\sum_{r=R+1}^{\infty}\lambda^0_{ir}f^0_{tr}$, and the factor projection method will just project out idiosyncratic noise,
or even contributions from $X_{it}^\prime \, (\widehat \beta_{\rm LS}-\beta^0)$.
This implies that the problematic variation associated with $\lambda^0_{ir}f^0_{tr}$ for most singular values $r$ remains.
This explains why it is so difficult to show anything
better than the convergence rate results in
Theorem~\ref{cor:LSconsistency} for the estimator 
$\widehat \beta_{\rm LS}$ in our setting.

\section{Asymptotic results for the group fixed-effect estimator}\label{sec:GroupedFE}

The main goal of this section is to derive asymptotic results for
the estimator $\widehat \beta_{\rm GS}$ defined in \eqref{est:clusterFEsplit},
which is the sample-splitting version of the group fixed-effect estimator. 
But we are first going to discuss the 
initial group fixed-effect estimator $\widehat \beta_{\rm G}$ defined
in \eqref{est:clusterFE} without sample-splitting. We will not actually derive
convergence rate results for $\widehat \beta_{\rm G}$ itself, but the discussion
of the approximation bias of $\widehat \beta_{\rm G}$ will be a very useful
precursor of the results for $\widehat \beta_{\rm GS}$.

\subsection{Results for  $\widehat \beta_{\rm G}$}\label{sec:GroupedFEconsistency}

We can rewrite our estimating equation for the group fixed-effect estimator in \eqref{TwoWayFEbasic} as
\begin{align}
   Y &=   X \cdot \beta + \delta \, D_{\delta}' + D_{\nu} \, \nu'  + \varepsilon  ,
\end{align}
where $\delta$ and $\nu$ are the $N \times C$ and $T \times G$ matrices of nuisance parameters,
while $D_{\delta}$ and $D_{\nu}$ are $T \times C$ and $N \times G$ are binary matrices in which each row contains a single one, indicating the group membership of 
the corresponding unit or time period, respectively.
By standard partitioned regression results we can then rewrite the group fixed-effect estimator in \eqref{est:clusterFE} 
as
\begin{align}
         \widehat \beta_{\rm G}  &= \left( \sum_{i = 1}^N \sum_{t = 1}^T \widetilde{X}_{it}^\prime \, \widetilde{X}_{it}\right)^{-1} \sum_{i = 1}^N \sum_{t = 1}^T \widetilde{X}_{it}^\prime \, \widetilde{Y}_{it} ,
         &
      \widetilde{X}_k &= M_N \, X_k \, M_T  , 
      &
       \widetilde{Y} &= M_N \,Y \, M_T ,  
       \label{gFEtilde}
\end{align}
where $\widetilde{X}_{it} = \left(\widetilde{X}_{it,1},\ldots,\widetilde{X}_{it,K} \right)$,
 $\widetilde{Y}_{it}$ and $\widetilde{X}_{it,k}$ are the entries of the $N \times T$ matrices $\widetilde{X}_k$
and $\widetilde{Y} $, respectively, and  $M_N = \mathbb{I}_N - D_{\nu}(D_{\nu}^\prime D_{\nu}) ^{-1}D_{\nu}^\prime$ and $M_T = \mathbb{I}_T - D_{\delta}(D_{\delta}^\prime D_{\delta}) ^{-1}D_{\delta}^\prime$ are projection matrices of dimesion
$N \times N$ and $T \times T$, respectively.

Using this representation of the group fixed-effect estimator and the model in \eqref{eq:consistency.model:general} we obtain that
\begin{align}
   \widehat \beta_{\rm G} - \beta^0 &= \phi_{NT} + \kappa_{NT},
   \label{DecomposeBetaG}
\end{align}
where
\begin{align}
   \phi_{NT} & := \left( \sum_{i = 1}^N \sum_{t = 1}^T \widetilde{X}_{it}^\prime \widetilde{X}_{it}\right)^{-1} \sum_{i = 1}^N \sum_{t = 1}^T \widetilde{X}_{it}^\prime \,  \varepsilon_{it},
  &
  \kappa_{NT} &:= \left( \sum_{i = 1}^N \sum_{t = 1}^T \widetilde{X}_{it}^\prime \widetilde{X}_{it}\right)^{-1} \sum_{i = 1}^N \sum_{t = 1}^T \widetilde{X}_{it}^\prime \, \widetilde \Gamma_{it},
  \label{DefPhiKappa}
\end{align}
with $\widetilde \Gamma$ defined analogously to $\widetilde{X}_k$ and $\widetilde{Y}$ in \eqref{gFEtilde}.
In the definition of $\phi_{NT}$  we can 
equivalently write $\widetilde \varepsilon_{it}$ instead of 
$\varepsilon_{it}$, but since $M_N$ and $M_T$ are idempotent matrices,
and $\widetilde{X}_{it}$ is already the projected regressor, this does not matter.
The same is true, of course, for $\widetilde \Gamma_{it}$ vs $\Gamma_{it}$
in the definition of $\kappa_{NT}$. However, the expressions in
\eqref{DefPhiKappa} turn out to be convenient as written.

Here, $\kappa_{NT}$ is the approximation error of having replaced
the nonlinear specification $\Gamma_{it} = \fct(\alpha_i,\gamma_t)$
in our model in \eqref{model:general} by the much
simpler additive specification 
$\delta_{i,c_t} +  \nu_{t,g_i}$ in the estimation equation \eqref{TwoWayFEbasic}.
To see this, we can
use standard matrix inequalities to bound the Euclidian norm
of $\kappa_{NT}$ by
\begin{align}
      \left\|  \kappa_{NT} \right\|
    &\leq  \left\| \left(\sum_{i = 1}^N \sum_{t = 1}^T \widetilde{X}_{it}^\prime \widetilde{X}_{it}\right)^{-1} \right\|
   \left( \max_{k} \left\| \widetilde{X}_{k} \right\|_F \right)
     \, \big\| \widetilde \Gamma \big\|_F,
\label{ExplainApproxError1}     
\end{align}
where $\|\cdot \|_F$ refers to the Frobenius norm.
Due to the definition of $M_N$ and $M_T$ we have
\begin{align}
   \big\| \widetilde \Gamma \big\|^2_F = \min_{\delta \in \mathbb{R}^{N \times C}} \,
    \min_{\nu \in \mathbb{R}^{T \times G}} \,
    \sum_{i=1}^{N}\sum_{t=1}^{T}\left[ \fct(\alpha_i,\gamma_t) - \delta_{i,c_t} -  \nu_{t,g_i} \right]^2  .
\label{ExplainApproxError2}    
\end{align}
The last two displays show that $\kappa_{NT}$ is small 
whenever $ \fct(\alpha_i,\gamma_t)$ can be well approximated by
$\delta_{i,c_t} +  \nu_{t,g_i}$. 
In equation \eqref{eq:taylorApprox} we already informally discussed the
magnitude of this  
approximation error, and found that it is of order $\|\alpha_i -\overline \alpha_{g_i} \|^2 + \|\gamma_t - \overline \gamma_{c_t} \|^2$. We now want to provide a more formal discussion of this and show that $\kappa_{NT}$ is
asymptotically small under appropriate regularity conditions.

In Section~\ref{Sect:clusterFE} we described the clustering 
algorithms that delivers the group memberships $g_i$ and $c_t$
based on the initial estimates $\widehat{\lambda}^*_i$ and $ \widehat{f}^*_t$.
The goal of the clustering is to group units $i$ with approximately 
the same value of $\alpha_i$, and to group time periods $t$ with
approximately the same $\gamma_t$. It is therefore crucial
that $\widehat{\lambda}^*_i$ and $ \widehat{f}^*_t$ are good
proxies for $\alpha_i$ and $\gamma_t$. Specifically, we require that there exist functions $\lambda^* : {\cal A} \rightarrow \mathbb{R}^{R_*}$ and $f^* : {\cal C} \rightarrow \mathbb{R}^{R_*}$
such that $\widehat{\lambda}^*_i$ and $ \widehat{f}^*_t$ converge
to the non-random limits $\lambda^*(\alpha_i) $ and $f^*(\gamma_t)$
as $N,T \rightarrow \infty$. The following assumption formalizes this and states 
all the regularity condition that we require on 
$\fct(\cdot,\cdot)$,  $\lambda^*(\cdot) $, $f^*(\cdot)$, 
$\widehat{\lambda}^*_i$, $ \widehat{f}^*_t$, and $X_{it}$.

\begin{assumption} 
    \label{ass:ApproxError}
   There exists a sequence $\xi_{NT}>0$ such that $\xi_{NT}  \to 0$ as $N,T \to \infty$, and
    \begin{enumerate}[(i)]     

        \item The function $\fct(\cdot,\cdot)$ is  at least twice continuously differentiable with uniformly bounded second derivatives.
        
        \item Every unit $i$ is a member of exactly one 
        group $g_i \in \{1,\ldots,G\}$, and every time period $t$ is 
        a member of exactly one group $c_t \in \{1,\ldots,C\}$. The size
        of all $G$ groups of units, and the size
        of all $C$ groups of time periods is 
        bounded uniformly by $Q_{\max}$.
        
        \item There exists $B>0$ such that 
         $\left\|a - b \right\| \leq B \left\| \lambda^*(a) -  \lambda^*(b) \right\|$ for all $a,b\in {\cal A} $
        , and  $\left\| a - b \right\| \leq B \left\| f^*(a) -  f^*(b) \right\|$ for all $a, b \in {\cal C}$, and the domains 
        ${\cal A}$ and ${\cal C}$ are convex set.
              
       \item
       $ \frac 1 {N}  \sum_{i=1}^N \left(   \left\| \widehat{\lambda}^*_i - \lambda^*(\alpha_i) \right\|^2 
         \right)= O_P \left( \xi_{NT} \right) $, \;
   $ \frac 1 {T}  \sum_{t=1}^T \left(   \left\| \widehat{f}^*_t - f^*(\gamma_t) \right\|^2 
   \right)= O_P \left(  \xi_{NT} \right) $.
   
       \item 

       $ \frac 1 {N}  \sum_{i=1}^N   \left\| \widehat{\lambda}^*_i - \widehat{\lambda}^*_{j(i)} \right\|^2 = O_P \left(  \xi_{NT} \right) $ for any matching function $j(i) \in \{1,\ldots,N\}$ such that $g_i = g_{j(i)}$,
      and
      $\frac 1 {T}  \sum_{t=1}^T  \left\| \widehat{f}^{\,*}_t - \widehat{f}^{\,*}_{s(t)} \right\|^2 = O_P \left(  \xi_{NT} \right) $ for any matching function $s(t) \in \{1,\ldots,T\}$ such that $c_t = c_{s(t)}$.

        \item 
        $\max_{k,i,t} \left| \widetilde{X}_{it,k} \right| = O_P(1)$,
and        
        $\plim_{N,T \rightarrow \infty}  \frac 1 {{NT}} \sum_{i=1}^N \sum_{t=1}^T \widetilde X_{it}^\prime \widetilde X_{it} = \Omega$, where $\Omega$ is a positive definite non-random matrix.

    \end{enumerate} 
\end{assumption}

\begin{lemma}
   \label{lemma:ApproxError}
   Under Assumption~\ref{ass:ApproxError} we have
   $$
      \kappa_{NT} = O_P(\xi_{NT})
   $$
\end{lemma}

The lemma shows that the approximation error $\kappa_{NT}$ vanishes
at  rate $\xi_{NT}$ as $N,T \rightarrow \infty$. The assumption
and  lemma are formulated for arbitrary rates, but
as will become clear from the following discussion,
the best we can
achieve in our setting is a rate of $\xi_{NT} = 1/\min(N,T)$,
which coincides with $\xi_{NT} = 1/\sqrt{NT}$ in the special case
that $N$ and $T$ grow at the same rate.

Part (i) of Assumption~\ref{ass:ApproxError} requires the
function $\fct(\cdot,\cdot)$ to be sufficiently smooth. This condition should 
not be surprising, because our informal
discussion of the approximation error in equation~\eqref{eq:taylorApprox} already relies on a second order Taylor expansion
of $\fct(\cdot,\cdot)$, and the proof of Lemma~\ref{lemma:ApproxError} is
based on exactly such an expansion.

Part (iii) and (iv) of the assumption
are analogous to ``Assumption 2 (injective moments)'' in 
\cite{bonhomme2017discretizing}, except that they consider
a one-way fixed effect setting while we consider a two-way fixed effect setting. Part (iii) requires the functions $\lambda^*(\cdot) $ and $f^*(\cdot)$
to be injective, that is, $\alpha_i$ and $\gamma_t$ can be 
uniquely recovered from knowing $\lambda^*(\alpha_i) $ and $f^*(\gamma_t)$.
A necessary condition for this is that 
\begin{align}
R^* \geq \max(d_\alpha,d_\gamma), 
   \label{lowerBoundRstar}
\end{align}
where $d_\alpha$ and $d_\gamma$ are the dimensions of $\alpha_i$
and $\gamma_t$, respectively.
Part (iv) requires the estimates $\widehat{\lambda}^*_i$ and $ \widehat{f}^*_t$ to converge to  $\lambda^*(\alpha_i) $ and $f^*(\gamma_t)$ at the average rate of $\xi_{NT}^{1/2}$.  
We expect that the estimated eigenfunctions of $\fct(\alpha_i,\gamma_t)$, which correspond to the estimated factor loadings and factors, proposed as cluster proxies in Section~\ref{Sect:clusterFE} satisfy this assumption by an application of Theorem 1 from \cite{BaiNg2002}. 
Since $T$ observations are available for unit $i$
we expect that $\widehat{\lambda}^*_i$ converges at a rate of
$T^{1/2}$, and since $N$ observations are available for time period $t$
we expect that $ \widehat{f}^*_t$ converges at a rate of
$N^{1/2}$, see also, for example, Theorem~1 and~2 in \cite{bai2003inferential}.
This explains why $\xi_{NT} = 1/\min(N,T)$ is the best rate
we can achieve here.

Part (v) of Assumption~\ref{ass:ApproxError} is a high-level assumption on the clustering mechanism used to obtain the group memberships $g_i$
and $c_t$. For units $i$ and $j$ in the same group,
and for time periods $t$ and $s$ in the same group, we demand the average
differences $\widehat{\lambda}^*_i - \widehat{\lambda}^*_j$
and $\widehat{f}^{\,*}_t - \widehat{f}^{\,*}_{s}$ to be small as $N,T \rightarrow \infty$. 
In other words, we require that the clustering mechanism does what it is
intended to do, namely forming groups such that the estimates 
 $\widehat{\lambda}^*_i$ and $\widehat{f}^{\,*}_t$ 
 for units $i$ and time periods $t$ in the same group
 are close to each other. For a given clutstering algorithms (e.g.\
 the one describe in Section~\ref{Sect:clusterFE}) one could prove that this
 assumption  holds under further regularity conditions on the distribution of 
 $\alpha_i$ and $\gamma_t$, see, for example,
 Lemma~1 in \citealt{bonhomme2017discretizing}. 
 In particular, a necessary condition for
 part (v) of Assumption~\ref{ass:ApproxError} to hold is the following:
 
\begin{condition} 
  $ \frac 1 {N}  \sum_{i=1}^N   \left\| \alpha_i - \alpha_{j(i)} \right\|^2 = O_P \left(  \xi_{NT} \right) $ for any matching function $j(i) \in \{1,\ldots,N\}$ such that $g_i = g_{j(i)}$,
      and
      $\frac 1 {T}  \sum_{t=1}^T  \left\| \gamma_t - \gamma_{s(t)} \right\|^2 = O_P \left(  \xi_{NT} \right) $ for any matching function $s(t) \in \{1,\ldots,T\}$ such that $c_t = c_{s(t)}$.
\end{condition}     
 
This condition coincides with  Assumption~\ref{ass:ApproxError}(v)
in the unrealistic case that 
$\widehat{\lambda}^*_i = \alpha_i$ and
$\widehat{f}^{\,*}_t = \gamma_t$. Starting from this 
unrealistic case and then applying the transformations
$\lambda^* : {\cal A} \rightarrow \mathbb{R}^{R_*}$
and $f^* : {\cal C} \rightarrow \mathbb{R}^{R_*}$
and adding noise to the estimates then gives part (v) of Assumption~\ref{ass:ApproxError}.
Crucially, for this regularity condition to hold, we need that  
$\xi_{NT} \gtrsim 1/\min(N^{2/d_\alpha},T^{2/d_\gamma})$,
see Lemma~2 in \cite{bonhomme2017discretizing} for the analogous results in a one-way fixed effect 
model (also \citealt{graf2002rates}). Since our actual clustering method is not
based on the unobserved $\alpha_i$ and $\gamma_t$, but on
$\widehat{\lambda}^*_i$ and 
$\widehat{f}^{\,*}_t $ we require the stronger condition
(in view of \eqref{lowerBoundRstar}) that
$$
  \xi_{NT} \gtrsim [\min(N,T)]^{-2/R^*} .
$$
This is a necessary condition 
for Assumption~\ref{ass:ApproxError}(v) to be satisfied.\footnote{%
Following the logic in \cite{bonhomme2017discretizing} we believe that
we actually only need
   $\xi_{NT} \gtrsim 1/\min(N^{2/d_\alpha},T^{2/d_\gamma})$, that is, our group fixed effect
   estimator $\widehat \beta_{\rm G}$ truly cannot achieve a convergence rate
   faster than $1/\min(N^{2/d_\alpha},T^{2/d_\gamma})$. Thus,
   if $R^* > \max(d_\alpha,d_\gamma)$, then 
   $\xi_{NT} \gtrsim [\min(N,T)]^{-2/R^*}$ is probably not a necessary 
   condition for the result of Lemma~\ref{lemma:ApproxError} itself, 
   but only for our Assumption~\ref{ass:ApproxError}(v).
}
Therefore, if we want to achieve the best possible rate
$\xi_{NT} = 1/\min(N,T)$, then we need $R^* \leq 2$, which according
to \eqref{lowerBoundRstar} implies that $d_\alpha \leq 2$
and $d_\gamma \leq 2$.
This discussion shows that our group fixed-effect estimator $\widehat \beta_{\rm G}$
 suffers from a curse of dimensionality with regards to the
dimensions of $\alpha_i$ and $\gamma_t$. However, this should be unsurprising, given
the semi-parametric nature of the estimation problem -- with non-parametric component
$\fct(\alpha_i,\gamma_t)$. This also shows that there is a tradeoff between the LS
estimator analyzed in Section~\ref{sect:consistencyBAI} and the group fixed effects estimator discussed here -- we will further compare those two estimators in our MC analysis below.

Finally, part (vi) of Assumption~\ref{ass:ApproxError} requires
some regularity conditions on the projected regressors
 $\widetilde{X}_k = M_N \, X_k \, M_T$ 
 defined in \eqref{gFEtilde}.

This concludes our discussion of the approximation 
error $\kappa_{NT}$. We have argued that,
under appropriate regularity conditions, including
$\max(d_\alpha,d_\gamma) \leq 2$, 
we can use Lemma~\ref{lemma:ApproxError} 
to obtain $\kappa_{NT} = 1 / \sqrt{NT}$,
for $N$ and $T$ growing to infinity at the same rate.
Since $\widehat \beta_{\rm G} - \beta^0 = \phi_{NT} + \kappa_{NT}$
we could then conclude that 
$\widehat \beta_{\rm G} - \beta^0 = O_P(1 / \sqrt{NT})$,
if we could also show that $\phi_{NT}  = O_P(1 / \sqrt{NT}) $.

From the definition of $\phi_{NT}$ in \eqref{DefPhiKappa} one might think
that it is easy to derive this result on $\phi_{NT}$ by imposing
an approximate exogeneity condition on the regressors. 
However, the problem is that $\widetilde{X}_k$ depends on the 
group assignments of units $i$ and time periods $t$, which 
were constructed based on $\widehat{\lambda}^*$ and $ \widehat{f}^*$,
which depend on the errors $\varepsilon$. 
Thus, $\widetilde{X}_k$ depends on $\varepsilon$ in complicated ways 
through the group assignment, making a proof of 
 $\phi_{NT}  = O_P(1 / \sqrt{NT}) $ technically challenging.
In principle, we expect that
 \begin{align}
    \sqrt{NT} \, \phi_{NT} \Rightarrow {\cal N}(0,\Sigma_{\rm G})
    \label{AsymptPhiG}
 \end{align}
 holds for an appropriate covariance matrix $\Sigma_{\rm G}$, 
 and our simulations evidence suggest that this is indeed the case.
 However,
 we are not aiming to prove this result in this paper.
As explained already in Section~\ref{sec:estimators},
this technical difficulty in analyzing $\widehat \beta_{\rm G}$ is exactly why we introduced the
split-sample version of the group fixed-effect estimator, for which
we are going to derive results in the following.

\subsection{Results for  $\widehat \beta_{\rm GS}$}\label{sec:GroupFEsplitnormal}

The split-sample version of the group fixed effect estimator
 was introduced in Section~\ref{sec:DefinedBetaGS} above.
Using the Frisch-Waugh-Lovell theorem we can rewrite 
$\widehat \beta_{\rm GS}$ in equation \eqref{est:clusterFEsplit} as follows:
\begin{align*}
   \widehat \beta_{\rm GS}
   = \left( \sum_{s=1}^4
   \sum_{(i,t) \in  {\cal O}_{s}}
   \widetilde{X}_{it}^{(s)\,\prime} \widetilde{X}_{it}^{(s)}\right)^{-1}  \sum_{s=1}^4  \sum_{(i,t) \in  {\cal O}_{s}}   \widetilde{X}_{it}^{(s)\,\prime} \,  Y_{it} ,
\end{align*}
where the projected regressors
$\widetilde{X}_{it}^{(s)} = \left(\widetilde{X}_{it,1}^{(s)} , \ldots, \widetilde{X}_{it,K}^{(s)} \right)'$
for each
subpanel $s \in \{1,2,3,4\}$, each regressor
$k=1,\ldots,K$,
and observations $(i,t) \in  {\cal O}_{s}$ within that subpanel,
are the residuals of the least-squares
problem
\begin{align}
    \min_{\delta} \,
    \min_{\nu} \,
     \sum_{(i,t) \in  {\cal O}_{s}} \left( X_{it,k} -
    \delta_{i,c^{(s)}_t} -  \nu_{t,g^{(s)}_i}
    \right)^2 .
    \label{LSprojection}
\end{align}
Following the decomposition of 
$\widehat \beta_{\rm G}$ in \eqref{DecomposeBetaG}, we can now introduce
the analogous decomposition for $\widehat \beta_{\rm GS}$ by
\begin{align}
   \widehat \beta_{\rm GS} - \beta^0 &=  
    \phi^{\rm (GS)}_{NT} + \kappa^{\rm (GS)}_{NT}  ,
    \label{DecomposeBetaGS}
\end{align}
where
\begin{align*}
   \phi^{\rm (GS)}_{NT} & := \left( \sum_{s=1}^4
   \sum_{(i,t) \in  {\cal O}_{s}}
   \widetilde{X}_{it}^{(s)\,\prime} \widetilde{X}_{it}^{(s)}\right)^{-1}  \sum_{s=1}^4  \sum_{(i,t) \in  {\cal O}_{s}}   \widetilde{X}_{it}^{(s)\,\prime} \,  \varepsilon_{it},
  \\
  \kappa^{\rm (GS)}_{NT} &:= \left( \sum_{s=1}^4
   \sum_{(i,t) \in  {\cal O}_{s}}
   \widetilde{X}_{it}^{(s)\,\prime} \widetilde{X}_{it}^{(s)}\right)^{-1} \sum_{s=1}^4  \sum_{(i,t) \in  {\cal O}_{s}}      \widetilde{X}_{it}^{(s)\,\prime} \, \widetilde \Gamma^{(s)}_{it},
\end{align*}
Here, $\phi^{\rm (GS)}_{NT}$ is a variance term that we 
will show to be unbiased and asymptotically normal,
and $\kappa^{\rm (GS)}_{NT}$ is the approximation 
error from having replaced $\fct(\alpha_i,\gamma_t)$ 
by the linear grouped fixed effect in the
estimation for $\widehat \beta_{\rm GS}$ in \eqref{est:clusterFEsplit}.
 The $\widetilde \Gamma^{(s)}_{it}$ are the 
residuals of the least-squares problem \eqref{LSprojection}
when $X_{it,k}$ is replaced by $\Gamma_{it}=\fct(\alpha_i,\gamma_t)$. 

For each of the four subpaneles $s \in \{1,2,3,4\}$,
 the discussion of the approximation error $\kappa^{\rm (GS)}_{NT}$ is identical to the discussion of the
 approximation error $\kappa_{NT}$ of $\widehat \beta_{\rm G}$,
 see, in particular, the bounds
 \eqref{ExplainApproxError1} and \eqref{ExplainApproxError2} above. It is therefore straightforward to
 obtain the analogue of Lemma~\ref{lemma:ApproxError}
 for the approximation error of the split-sample estimator.
 
\begin{lemma}
   \label{lemma:ApproxError2}
   Under Assumption~\ref{ass:ApproxErrorSplit}  (in appendix) we have
   $$
      \kappa^{\rm (GS)}_{NT} = O_P(\xi_{NT})
   $$
\end{lemma} 
Assumption~\ref{ass:ApproxErrorSplit} is stated in the appendix, but it is simply
a restatement of Assumption~\ref{ass:ApproxError}
for each subpanel $s \in \{1,2,3,4\}$. Those assumptions 
were discussed after Lemma~\ref{lemma:ApproxError} above.
In particular, the best possible convergence rate
we can hope for here is $\xi_{NT} = 1/\min(N,T)$, but 
that rate is only attainable for
$d_\alpha \leq 2$ and $d_\gamma \leq 2$.

The key difference between  $\widehat \beta_{\rm G}$
and $\widehat \beta_{\rm GS}$ is that for the
split-sample estimator we can derive the asymptotic behavior 
of the variance term very easily $\phi^{\rm (GS)}_{NT}$.
For this purpose, we impose the following assumption.

\begin{assumption}~
    \label{ass:AsympSplit}
    \begin{enumerate}[(i)]  
        \item  Conditional on $X$, $\alpha$, $\gamma$, we assume that $\varepsilon_{it}$ is independently distributed across $i$
        and over $t$, such that 
        $ \sigma^2_{it} := \mathbbm{E}\left( \varepsilon_{it}^2 \, \big| \, X,\alpha,\gamma\right) \leq B <\infty$, for some constant $B$
        that is independent of $i,t,N,T$.
        
        \item We have $\plim_{N,T \rightarrow \infty} \frac 1 {NT} \sum_{s=1}^4
   \sum_{(i,t) \in  {\cal O}_{s}}
   \widetilde{X}_{it}^{(s)\,\prime} \widetilde{X}_{it}^{(s)} = \Omega  >0$,
   and for each $s \in \{1,\ldots,S\}$ we have
   $\plim_{N,T \rightarrow \infty}\frac 1 {NT}  
   \sum_{(i,t) \in  {\cal O}_{s}}
   \sigma^2_{it}
   \widetilde{X}_{it}^{(s)\,\prime} \widetilde{X}_{it}^{(s)} = \Sigma^{(s)}$. Furthermore, we assume that, for $s \in \{1,2,3,4\}$, all the 
   third-order sample moments of
    $\widetilde{X}_{it}^{(s)\,\prime} \,  \varepsilon_{it}$
    across $(i,t) \in  {\cal O}_{s}$ are bounded
    as $N,T \rightarrow \infty$.
        
    \end{enumerate}
\end{assumption}

Assumption~\ref{ass:AsympSplit} 
together with the sample splitting method used to construct $\widehat \beta_{\rm GS}$
guarantees that,
within each subpabel $s \in \{1,2,3,4\}$,
the
$ \widetilde{X}_{it}^{(s)\,\prime} \,  \varepsilon_{it}$
are zero mean and independently distributed
across $ (i,t)   $.
Here, the split-panel construction is crucial, since
it guarantees that $\widetilde{X}_{it}^{(s)}$ 
is independent of $ \varepsilon_{it}$. 
The remaining conditions in Assumption~\ref{ass:AsympSplit}
are regularity conditions to allow us to apply
the Lyapunov central limit theorem for each subpanel
and to guarantee that $\phi^{\rm (GS)}_{NT}$ has a 
finite asymptotic variance. We therefore obtain the
following lemma.

 \begin{lemma}
   \label{lemma:CLT}
   Under Assumption~\ref{ass:AsympSplit} we have,
   as $N,T \rightarrow \infty$,
   \begin{align*}
      \sqrt{NT} \, \phi^{\rm (GS)}_{NT}   &\Rightarrow {\cal N}(0, \Sigma_{\rm GS}) ,
      &
      \Sigma_{\rm GS} &= \Omega^{-1} \left( \sum_{s=1}^4 \Sigma^{(s)}
      \right) \Omega^{-1} .
   \end{align*}
\end{lemma}

Combining equation \eqref{DecomposeBetaGS}
with Lemma~\ref{lemma:ApproxError2} and Lemma~\ref{lemma:CLT}
then gives the following theorem.

\begin{theorem}
   Under Assumption~\ref{ass:AsympSplit} and Assumption~\ref{ass:ApproxErrorSplit} we have
   $$ \widehat \beta_{\rm GS} - \beta^0 =  
    O_P\left( \frac 1 {\sqrt{NT}} + \xi_{NT} \right)
    = o_P(1).$$
\end{theorem}

Analogous to Corollary~\ref{cor:LSconsistency}
for the least-squared estimator of \cite{Bai2009},
we have this obtained a consistency result for 
$\widehat \beta_{\rm GS}$ as well. We have not derived
asymptotic inference results using either of these estimators,
but in the following section we explain how we use those
estimators to construct confidence intervals in our simulations
and empirical application.

\section{Implementation}
\label{sec:Implementation}

The asymptotic results derived for
$\widehat\beta_{\rm LS} $, $\widehat\beta_{\rm G} $,
and $\widehat\beta_{\rm GS} $ in the last two sections
are insightful for how those estimates should be used
in practice. In particular, our discussions and derivations
are helpful to appreciate the limitations 
and assumptions needed for the estimation approaches,
and we will summarize those again in our conclusion section below.

In the following Monte Carlo simulations and empirical application we will employ the estimates $\widehat\beta_{\rm LS} $, $\widehat\beta_{\rm G} $,
and $\widehat\beta_{\rm GS} $ in a way that goes beyond our  formal asymptotic results. In particular, we will
use all those estimators to construct confidence intervals
and we will also apply Jackknife methods for bias correction.
In this section, we want to briefly explain how those
confidence intervals and bias corrected estimates are constructed. 

To calculate standard errors for each estimator we  ignore the approximation error discussed in our formal results and simply use formulas as if residuals were independently distributed. For example, in section \ref{sec:GroupedFEconsistency} where we split the residual term into $\phi$ and $\kappa$, we will ignore the $\kappa$ term and estimate standard errors as if we are left with only $\phi$. We use the jackknife corrections to address the residual terms related to approximation error in both the factor and grouped fixed-effects estimation models.

For factor model standard errors we construct the heteroscedasticity-consistent estimator from \cite{white1980heteroskedasticity} as follows. Take $\Omega = \sum_{i=1}^N\sum_{t=1}^T \widetilde X_{it}^\prime \widetilde X_{it}$ and $\widehat\Sigma = \sum_{i=1}^N\sum_{t=1}^T \widehat u_{it}^2  \widetilde X_{it}^\prime \widetilde X_{it}$ where $\widehat u_{it} = \widetilde Y_{it} - \sum_k \widehat\beta_{LS,k} \widetilde X_{it,k}$ and for a matrix $A$, in this context, $\widetilde A$ represents the matrix with factors projected. We must make a degrees of freedom correction for the factor projection by the ratio $dfc = \sqrt{\frac{NT}{(N-R)(T-R)}}$. Then the vector of standard errors are, 
\begin{align*}
    {\rm se}(\widehat \beta_{\rm LS}) = dfc \cdot \sqrt{
    {\rm diag}\left(\Omega^{-1} \widehat\Sigma \Omega^{-1}\right)}. 
\end{align*}
As above, we use this same standard error estimator for jackknife corrected estimates. 

For the grouped fixed-effects models we use clustered standard errors where clusters are taken as the combination of $i$ and $t$ clusters. 
That is, for the matrices of clusters $D_\alpha$ and $D_\gamma$ for $i$ and $t$ respectively we take clusters as the Kronecker product between these two matrices, $D_\alpha\otimes D_\gamma$. 
Remember here that the columns of $D_\alpha$, resp. $D_\gamma$, are the cluster assignments of $i$, resp. $t$ with a 1 entry if that observation is in the cluster and a 0 otherwise. 
Take $m$ as the index for cluster assignment with $M = GC$ the total number of clusters. Hence, $D_\alpha\otimes D_\gamma := \mathcal{D}$ is an $NT$ by $M$ matrix with $\mathcal{D}_m$ representing a column of this matrix and $\mathcal{D}_{n,m}$ representing an entry. 
A combination $(i,t)$ can be identified by the row, $n$, of the matrix $\mathcal{D}$ as $t= \ceil{n/N}$ and $i = n - \left(\ceil{n/N} - 1\right)N$, which is similar to the usual matrix flattening procedure. 
Then, the column-vector $\mathcal{D}_m$ consists of a 1 if the $(i,t)$ combination implied by that row, $n$, is in that column's cluster and 0 otherwise. 

Define as above $\Omega = \sum_{i=1}^N\sum_{t=1}^T \widetilde X_{it}^\prime \widetilde X_{it}$ and $\widehat u_{it} = \widetilde Y_{it} - \sum_k \widehat\beta_{G,k} \widetilde X_{it,k}$ where in this context for matrix $A$, the matrix $\widetilde A$ represents the matrix with group fixed-effects projected out. Call the index function $n(i,t) = i + (t-1)N$, such that $\mathcal{D}_{n(i,t), m}$ returns the binary indicator of whether $(i,t)$ is in the $m^{\textrm{th}}$ combination cluster. Now define $\widehat \Sigma = \sum_{m = 1}^M \sum_{i=1}^N\sum_{t=1}^T \mathcal{D}_{n(i,t), m} \widehat u_{it}^2  \widetilde X_{it}^\prime \widetilde X_{it}$. This collapses the familiar block-diagonal matrix where values within each block corresponds to a combination cluster and are unrestricted but zero outside each block. 
The clustered standard errors can thus be defined as 
\begin{align*}
    {\rm se}(\widehat \beta_{G}) = dfc \cdot \sqrt{ {\rm diag}\left(\Omega^{-1} \widehat\Sigma \Omega^{-1}\right)}
\end{align*}
where in this context $dfc = \sqrt{\frac{NT}{(N - G)(T - C)}}$. The standard error estimator is identical for the split sample version except there are many more combination clusters by the nature of this split sample estimators clustering method.

Finally, in our Monte Carlo simulations below we also explore
 whether Jackknife bias correction methods are able to reduce the approximation bias
and the incidental parameter bias of the various estimates.
We do not have any theoretical results on the leading order
bias of the various estimates, but we nevertheless we
 follow \cite{fernandez2016individual} to estimate the jackknife bias corrected analog to each estimator as follows. This procedure is closely related to \cite{Dhaene2015split}. First, split the sample along the $i$ dimension into two $N/2$ by $T$ samples. For each of these samples run and call the related estimates from estimator $E$, $\widehat\beta^{1,1}_E$ and $\widehat\beta^{1,2}_E$, respectively. Repeat this process along the $t$ dimension to return $\widehat\beta^{2,1}_E$ and $\widehat\beta^{2,2}_E$. Then the final jackknife bias corrected analog for estimator $E$ is 
\begin{align*}
    \widehat\beta_{E,JK} = 3\widehat\beta_E - \frac{1}{2}\left( \widehat\beta^{1,1}_E + \widehat\beta^{1,2}_E\right) - \frac{1}{2}\left( \widehat\beta^{2,1}_E + \widehat\beta^{2,2}_E\right), 
\end{align*}
where $\widehat\beta_E$ is simply the estimate without any sample split. We maintain the assumption that standard errors are the same across split samples so we can simply take the standard error estimate from the whole sample.

\section{Monte Carlo simulations}\label{sect:simulations}

For our Monte Carlo simulations,  we choose
a data generating process with a single regressor ($K=1$), 
and we generate outcomes and regressor as follows:
\begin{align}
    \begin{split}
        Y_{it} &= X_{it}\beta^0 + \fct(\alpha_i, \gamma_t) + \varepsilon_{it} ,\\
        X_{it} &= \FCT(\alpha_i, \gamma_t) + \mu_{it} ,
    \end{split}
\end{align}
with 
\begin{align}
    \begin{split}
        \varepsilon_{it}, \, \alpha_{i}, \, \gamma_{t}, \,\mu_{it} \, &\sim 
        \text{all mutually independent and i.i.d.} \,        
        {\cal N}(0,1)
    \end{split}
\end{align}
This setting assumes that the endogeneity in $X_{it}$ depends on the specification of $\FCT(.,.)$ vis--\`a--vis $\fct(.,.)$. The decay in singular values for either the unobserved term in $Y_{it}$ or for $X_{it}$ can be directly manipulated through the specification of $\fct(.,.)$ and $\FCT(.,.)$, which will dictate the number of significant factors in each decomposition.

We set $\beta^0 = 1$ and,
\begin{align}
\begin{split}
    \fct(a,b) &= \FCT(a,b) = \frac{1}{\sqrt{2\pi}\theta} \exp\left(- \frac{(a - b)^{2}}{\theta^2}\right), \quad \theta = (1/2)^3.
\end{split}
\end{align}
The $\theta$ value here dictates the speed of decay in singular values for $\fct(.,.)$ and $\FCT(.,.)$, holding fixed the variation in their arguments, where a lower value implies a slower decay. This particular value for $\theta$ was chosen as it implies a slow decay in singular values such that the endogenous component of the unobserved term and $X$ persists even as many factors are included. The value for $\theta$ carries no fundamental economic meaning. Note, the nature of bias in this simulation is by design monotonic and positive for illustrative purposes.  

Table \ref{tab:sim.small.output} below shows the results from 10,000 Monte Carlo simulations. 
These results display our theoretical result on bias reduction succinctly. 
We see that as we increase the number of factors the average bias reduces and the standard deviation of estimates increases. 
We also see a significant improvement in bias using the grouped fixed-effects estimator, without a large increase in standard deviation. The GFE split sample estimator performs much worse in terms of bias, which is expected given the significantly smaller candidate pool for clustering in this estimator. 
The jackknife analog to each estimator reduces bias in all cases except the factor model with 5 factors, but significantly increases standard deviation in all cases. 
Note that we only report factor model estimated after first applying a within transformation, but we actually do not find any substantial difference compared to not applying the within transformation first.

\begin{table}[ht]
\centering
\caption{Monte Carlo simulations}
    \begin{tabular}{rrrrrrr}
  \hline
 & Bias & St. Dev. & Mean $\widehat{se}$  &  CDF($\beta^0$)
 & Cover &  MC cover\\ 
  \hline
  OLS & 0.6025 & 0.0094 & 0.0080 & 0.00  & 0\% & 0\% \\
  Fixed-effects & 0.5165 & 0.0095 & 0.0087& 0.00  & 0\% & 0\% \\
 LS (5 factors) & 0.0418 & 0.0114 & 0.0105& 0.00  & 4\% & 0\% \\ 
  LS (20 factors) & 0.0158 & 0.0147 & 0.0110 & 0.14 & 64\% & 28\% \\ 
  LS (50 factors) & 0.0129 & 0.0316 & 0.0135 & 0.34 & 56\% & 68\% \\ 
  LS Jackknife (5 factors) & -0.4216 & 0.0317 & 0.0105 & 1.00 & 0\% & 0\% \\ 
  LS Jackknife (20 factors) & -0.0080 & 0.0282 & 0.0110 & 0.61 & 54\% & 78\% \\
  LS Jackknife (50 factors) & -0.0042 & 0.0795 & 0.0135 & 0.52 & 26\% & 96\% \\ 
  Group fixed-effects & 0.0027 & 0.0176 & 0.0179 & 0.44 & 95\% & 88\% \\ 
  GFE jackknife & 0.0009 & 0.0321 & 0.0179 & 0.49 & 73\% & 98\% \\ 
  GFE splits & 0.0204 & 0.0184 & 0.0126 & 0.13 & 59\% & 27\% \\ 
   \hline
\end{tabular}

\vspace{1ex}
 {\raggedright {\footnotesize  \small{N = T = 100} with 10,000 repetitions.\par
 
 All results refer to estimation of $\beta$. 
 Bias is simply the mean of the bias across simulations. 
 Standard deviation is the standard deviation of the estimates, again across simulations. 
 Mean $\widehat{se}$ is the mean across simulations of the standard error estimate. 
  CDF($\beta^0$) is value of the empirical CDF across simulations evaluated at the true value of $\beta^0=1$. 
 Cover is defined here as the percentage of the
 95\% confidence intervals containing the true $\beta^0$. 
 MC cover reports coverage if estimates are normally distributed with mean bias from column 2 and standard deviation from column 3. 
 \par}}
    
    \label{tab:sim.small.output}
\end{table}

If we compare the mean standard error estimates to standard deviation across simulations we see evidence that the standard error calculation may underestimate the true standard error of the estimator. 
In light of discussion in Section~\ref{sec:Implementation}, we explicitly ignore fixed-effects approximation error and assumed only a noise term remains when estimating standard errors, which may explain this discrepancy. 
The divergence between estimated standard errors and standard deviation across simulations is particularly noticeable for the factor model with a large number of factors and for jackknife bias corrected estimators. 
For large factor models it is likely our inference approach misses out dependence structures introduced by the factor projection. 
This divergence is less pronounced for the group fixed-effects estimator without bias correction. 
It is also worth noting the assumption of equal standard errors across the components of the jackknife estimator appears to be violated from the difference in standard deviations between jackknife and non-jackknife estimators. 
These results suggest an alternative method, for example using bootstrap, is necessary to do feasible inference in this setting with these estimators.
Since we do not expressly advocate a particular inference approach for any estimator used in this paper we do not discuss this issue any further and leave it for future research. 

In Table~\ref{tab:sim.small.output} we also compare where the true value of $\beta^0 = 1$ lies in the empirical CDF of each estimator to the coverage based on a normal distribution with the simulated mean bias and standard deviation as the distribution parameters. 
We see that in instances where bias is low and $\beta^0 = 1$ is close to the median of the empirical CDF then $2|$CDF($\beta^0$) - 0.5$|$ $\times 100\%$ is approximately equal to 1 - MC cover. 
This is some evidence that the estimators may approach the normal distribution, where the simulations correctly estimate the standard deviations. 
However, given that the estimated standard errors are usually far from the simulation standard deviations, this still does not present a feasible inference procedure.

To compare the rates of convergence across estimators we repeat the above simulation exercise across different sample sizes, namely $N = T \in\{20,40,80,160\}$. 
The results are displayed in Table~\ref{tab:sim.sizes.output}. 
The table shows that for this range of data the convergence rates for the GFE estimators are all better or equal to the parametric rate.
Note, in this setting the parametric rate suggests the bias should halve for each increment in sample size. 
The factor model looks to be decaying at about the parametric rate, however, for the specification with a small number of factors (factors equal $\floor{3N^{1/4}}$ and $\floor{3N^{3/8}}$) the bias is substantially above standard deviation. 
This suggest there is a statistically significant persistence in bias for this estimator. 
For the factor model with $\floor{3N^{9/20}}$ factors, which is near the upper bound of number of factors, $O(\min\{N,T\}^{1/2})$, as per Theorem~\ref{thm:bai.con}, the bias does converge to within two standard deviations of zero. 
The standard deviations for each estimator do look to settle on the parametric rate by at the latest the last sample size increment, which is seen by comparison of the second last and last columns across estimators. In Appendix~\ref{app:laggedsimulations} we also include a simulation exercise with lagged dependent variables, which highlights the importants of having a correctly specified model.

\begin{table}[ht]
\centering
\caption{Convergence rate simulation}
    \begin{tabular}{rrrrr}
  \hline
 $N = T =$& 20 & 40 & 80 & 160   \\ 
  \hline
  Mean Bias &&&& \\
  (Standard Deviation) &&&& \\
  \hline
 LS ($\floor{3N^{1/4}}$ factors)& 0.4461 & 0.2782 & 0.2696  & 0.1609  \\ 
 & (0.0983) & (0.0405) & (0.0194) & (0.0097) \\ 
  LS ($\floor{3N^{3/8}}$ factors) & 0.3546  & 0.1748 & 0.0860 & 0.0177  \\ 
 & (0.1245) & (0.0388) & (0.0159) & (0.0067) \\ 
  LS ($\floor{3N^{9/20}}$ factors) & 0.2763  & 0.1077 & 0.0320 & 0.0122  \\ 
 & (0.1411) & (0.0386) & (0.0158) & (0.0071) \\ 
  Group fixed-effects & 0.2690  & 0.0064 & 0.0045 & 0.0008  \\ 
 & (0.1525) & (0.0545) & (0.0224) & (0.0110) \\ 
  GFE jackknife & 0.2458  & 0.0334 & 0.0023 & 0.0004  \\ 
 & (0.2225) & (0.0942) & (0.0406) & (0.0201) \\ 
  GFE splits & 0.3829  & 0.1524 & 0.0249 & 0.0036  \\ 
 & (0.1200) & (0.0657) & (0.0237) & (0.0111) \\ 
   \hline
\end{tabular}

\vspace{1ex}
 {\raggedright {\footnotesize  10,000 Monte Carlo rounds.\par
 
 All results refer to estimation of $\beta$. 
 Mean bias is simply the mean of the bias across simulations. 
 Standard deviation is the standard deviation of the estimates, again across simulations. 
 The multiple of 3 is applied here so the estimated number of factors is not too small for small sample sizes and such that there is a an actual change in the number of factors across the different sample sizes when $\floor{3N^{1/4}}$ factors are estimated. 
 \par}}
    
    \label{tab:sim.sizes.output}
\end{table}

\section{Empirical application}
\label{sec:empirical}

We apply our estimation procedure to an analysis of the UK housing market, following \cite{giglio2016no} (GMS16). Specifically, we study the effects of extremely long lease agreements on the price of housing, when compared to freehold agreements. In the UK housing market it is common for real estate property to be sold under each agreement. GMS16 posit that any change in price due to exogenous variation in whether the property was sold under extremely long lease or freehold must be attributed to so--called ``housing bubbles associated with a failure of the transversality condition''. The empirical challenge in making this comparison, and much discussed in GMS16, is to sufficiently control for observable and unobservable covariates such that variation in the variable of interest can be reasonably described as exogenous. 

In the following, we compare estimates using our method with the more flexible approach taken in their paper. We note first that given differences in data, these results should not be directly compared with GMS16.  Rather, this should be seen as an internal validity check across estimation models, i.e., to check if the aggregated setting produce similar estimates to the granular setting from GMS16 within the same set of data.

Consider the granular model from GMS16
\begin{align}\label{eqn:GMS16}
    Y_{iprt} = ExtremelyLongLease_i\beta + controls_{it}^\prime\delta + \phi_{prt} + \varepsilon_{iprt}
\end{align}
where $i$ are individual \textit{transactions} (i.e. not necessarily properties), $p$ is property type, $r$ are regions and $t$ is the month of transaction. Controls include hedonic variables, e.g. number of bedrooms, bathrooms and floorspace. $\phi_{prt}$ is a scalar fixed effect particular to the region, property type and month, and is identified via variation across transactions $i$. Compare this to an aggregated setting, 
\begin{align}\label{eqn:GMS16agg}
    Y_{rt} = ExtremelyLongLease_{rt}\beta + controls_{rt}^\prime\delta + \fct(\alpha_r,\gamma_t) + \varepsilon_{rt}
\end{align}
where $Y_{rt}$, $ExtremelyLongLease_{rt}$ and $controls_{rt}$ are the sample means aggregated to the region and transaction month.  The multidimensional array with entries $\phi_{prt}$ varies with higher rank than the matrix with entries $\fct(\alpha_r,\gamma_t)$ because the latter is constant across $p$ if extended to the equivalent multidimensional array with dimensions across $(p,r,t)$. This is why we believe the model in \eqref{eqn:GMS16} will better capture fixed-effects. 

 For purposes of this exercise we take the granular model with fixed-effects below as being, in theory, the better model to approximate unobserved heterogeneity. Hence we refer to this as the benchmark model. We use this benchmark approach to understand how well each estimator performs in practical instances where granular levels of aggregation are not always available, for example when data is aggregated for privacy reasons or for other feasibility reason. Hence, estimates close to the granular model estimates should be seen as performing ``well'' in this setting.

Table \ref{tbl:empirical} shows that when we control for fixed effects in the granular model there is a 0.3\% reduction in price when a long leasehold transaction is made compared to a freehold.
Whilst this is statistically significant, it translates to a decrease in the median house price of less than £1,000 so is arguably a small reduction economically. The OLS estimates do not change much across the different aggregation schemes and perhaps unsurprisingly the panel aggregated OLS has a much higher standard deviation due to the lower effective sample size. In the panel setting the factor model shows a convergence to the granular model with fixed effects as factors are increased and, interestingly, also to the grouped fixed-effects estimate, which is the closest to the benchmark estimates.\footnote{%
    In Table \ref{tbl:empirical}, our usual computation for the clustered standard errors of the group fixed-effect estimator was infeasible here due to the sample size.
    These standard error
    estimates are generated by resampling region clusters with replacement over 10,000 resamples. 
    }
 These results show a similar pattern to the simulation exercise where, according to the benchmark model, we see a bias reduction as the number of factors increases and when using the group fixed-effects estimator.

\begin{table}[H]
\centering
\caption{Empirical Results}
\begin{tabular}{r | rrr}
  \hline
 & Model & Estimate & Standard Errors  \\ 
 \hline
 Granular Model & Ordinary Least Squares & 0.203 & 0.0054 \\
 \eqref{eqn:GMS16} \hspace{0.65cm} & with Fixed Effects & -0.003 & 0.0006 \\
  \hline
  Panel Model & Ordinary Least Squares & 0.229 & 0.106\\
  \eqref{eqn:GMS16agg} \hspace{0.55cm} &LS factor model (5 factors)  & 0.024 & 0.012  \\ 
    &LS factor model (15 factors)  & 0.007 & 0.007  \\ 
    &LS factor model (30 factors)  & 0.007 & 0.008  \\ 
    &Group fixed-effects & 0.006 & 0.020 \\ 
   \hline
\end{tabular}

\vspace{1ex}
 {\raggedright {\footnotesize  UK housing market results for N = 2088 and T = 48.}}
\label{tbl:empirical}
\end{table}

\section{Conclusions}

Panel regressions are very popular estimation tools, 
because they allow to control for omitted variables that are unobserved and potentially correlated with the observed covariates.
Both \cite{Pesaran2006} and \cite{Bai2009}, and most of the literature following those
seminal papers, assume that those unobserved omitted variables
take the form of a low-rank matrix, which can be interpreted as
a static factor model or interactive fixed effects.
In this paper, we deviate from this interactive fixed effect model by assuming that the
unobserved omitted variables enter the model in the more general form
$\fct(\alpha_i,\gamma_t)$, where $h(\cdot,\cdot)$ is an unknown
smooth function, and $\alpha_i$ and $\gamma_t$ are (multidimensional)
fixed effects that can be arbitrarily correlated across $i$,
over $t$, and with the observed covariates.

We first explore the behavior of Bai's least-squares esimator
in this new setting.
We show that this LS esimator estimator is still consistent, as long as
the number of factors used in the estimation is allowed to grow
asymptotically.
However,  as explained in detail in Section~\ref{sect:consistencyBAI}, 
 it seems impossible to derive convergence rates faster 
than $ (\min\{N,T\})^{1/2}$ for this estimator in our setting.

We therefore  develop a new estimation approach called 
 the two-way grouped fixed effects approach, which generalize ideas
in \cite{bonhomme2017discretizing} to our two-way setting. 
We derive convergence rate results for the resulting new
 estimators and show that, depending on the dimension of $\alpha_i$
 and $\gamma_t$, and the relative size of $N$ and $T$, 
 convergence rates up to $\sqrt{NT}$ can be achieved with our new
 estimation approach.

We also explore the performance of those various estimators
in simulations and in an empirical application. We find that
both Bai's least-squares esimator and our grouped fixed effect
estimators tend to perform well in practice.
Interestingly, the theoretical convergence rate of $ (\min\{N,T\})^{1/2}$ for the LS esimator may 
often understate the performance of this esimator  
in practice. 

We also find that
Jackknife bias correction
helps to further reduce the bias of the 
various estimators, but at the cost
of increasing the variance.
Overall, the (Jackknife corrected) group fixed-effects estimator tends to have the smallest bias, but not necessarily the
smallest variance.
The empirical application shows that, according to our benchmark estimation, the LS estimation approach improves with more factors and that the group fixed-effects estimator does indeed provide a bias reduction compared to the LS estimator.

In the simulation exercise and empirical application we implemented standard error calculations for each estimator,
but we leave formal inference results in the setting of our paper as an open question for future research.

\setlength{\bibsep}{2pt}

\appendix

\section{Appendix}

\renewcommand{\theequation}{A.\arabic{equation}} 
\setcounter{equation}{0}
\renewcommand{\theassumption}{A.\arabic{assumption}} 
\setcounter{assumption}{0}
\renewcommand{\thelemma}{A.\arabic{lemma}} 
\setcounter{lemma}{0}

\subsection{Simulations with lagged dependent variable}\label{app:laggedsimulations}

In Table~\ref{tab:sim.lagged.output} we display the simulation results for the following DGP, 
\begin{align}
    \begin{split}
        Y_{it} &= Y_{i,t-1}\rho + X_{it}\beta + \fct(\alpha_i, \gamma_t) + \varepsilon_{it} ,\\
        X_{it} &= \FCT(\alpha_i, \gamma_t) + \mu_{it} ,
    \end{split}
\end{align}
where all parameters are set to the same values as Section~\ref{sect:simulations}, along with $\rho = 0.5$. 
Note that given we simulate $\gamma_t$ to be independent across $t$, there should be no omitted variable bias from simply ignoring $Y_{i,t-1}$ in the regression. 
However, and as we see in Table~\ref{tab:sim.lagged.output}, omitting $Y_{i,t-1}$ makes factor estimation more difficult because of the additional $Y_{i,t-1}\rho$ term in the fitted residual. 
To see this take the fitted residual without lagged $Y$ included, $\widehat{W}_1$, and with lagged $Y$ included, $\widehat{W}_2$, 
\begin{align*}
    \begin{split}
        \widehat{W}_1 &= (Y - X\widehat{\beta}) = X(\beta - \widehat{\beta}) + Y_{-1}\rho + \fct(\alpha, \gamma) + \varepsilon \\
        \widehat{W}_2 &= (Y - X\widehat{\beta} - Y_{-1}\widehat{\rho}) = X(\beta - \widehat{\beta}) + Y_{-1}(\rho - \widehat{\rho}) + \fct(\alpha, \gamma) + \varepsilon,
    \end{split}
\end{align*}
where $Y_{-1}$ is simply the matrix of lagged $Y$. 
We see then when lagged $Y$, or any control variable for that matter, is not included in the regression, then it makes identifying factors related to $\fct(\alpha, \gamma)$ more difficult due to $Y_{-1}\rho$ in the residual.
However, when lagged $Y$ is included in the regression then as long as $(\rho - \widehat{\rho})$ is reasonably small then factors are estimated better. 
Hence, whilst the presence of the lagged dependent term in the residual does not directly create omitted variable bias, it adds noise to the fitted residual and makes estimation of factors more difficult. 
This is especially highlighted by the fact that increased number of factors do not necessarily improve bias, from Table~\ref{tab:sim.lagged.output}. 
We see a similar, albeit less exaggerated, issue with the grouped fixed effects estimator, where cluster proxies for time periods are poorly estimated for the model without lagged $Y$ included in the regression. 

\begin{table}[H]
\centering
\caption{Lagged dependent variable simulation}
\begin{tabular}{rrr}
  \hline
    & Without lagged $Y$ & With lagged $Y$ \\ 
  \hline
 Mean bias (Standard deviation) & &    \\ 
  & &    \\ 
  \hline
  OLS & 0.9388 (0.0347) & 0.5625 (0.0102)  \\ 
    Fixed-effects & 0.4967 (0.0322) & 0.5109 (0.0096) \\ 
  LS (10 factors) & -0.0600 (0.0134) & 0.0191 (0.0124)  \\ 
  LS (20 factors) & -0.1445 (0.0134) & 0.0154 (0.0150)  \\ 
  LS (40 factors) & -0.2280 (0.0172) & -0.0384 (0.0831)  \\ 
  LS jackknife  (10 factors) & -0.1968 (0.0249) & -0.0380 (0.0229)  \\ 
  LS jackknife  (20 factors) & -0.2696 (0.0255) & -0.0089 (0.0286) \\ 
  LS jackknife  (40 factors) & -0.3296 (0.0401) &  -0.0582 (0.2418) \\ 
  GFE & -0.0354 (0.0248) &  0.0181 (0.0193) \\ 
  GFE jackknife & -0.0139 (0.0442) &  0.0157 (0.0343)  \\ 
   \hline
    
\end{tabular}

\vspace{1ex}
 {\raggedright {\footnotesize  10,000 Monte Carlo rounds.\par
 
 All results refer to estimation of $\beta$. 
 Mean bias is simply the mean of the bias across simulations. 
 Standard deviation is the standard deviation of the estimates, again across simulations. 
 \par}}
    
    \label{tab:sim.lagged.output}
\end{table}

\subsection{Proofs for Section~\ref{sect:consistencyBAI}}

 We first establish a technical lemma, which
 is afterwards used to prove the main text theorem.
Remember that we write $\|\cdot\|$ for the spectral norm
of a matrix.
Define the projection matrix $P_A = A(A^\prime A)^\dagger A^\prime$ for any matrix $A$ and remember we write the annihilation matrix $M_A = \mathbb{I} - P_A$. Here, $\dagger$ refers
to the Moore-Penrose inverse.

\begin{lemma}\label{lma:MWcons}
    Let Assumption \ref{ass:NC} hold
    and consider $N,T \rightarrow \infty$.    
    Furthermore,
    assume that
    \begin{align}
         Y =  \sum_{k=1}^K \, X_k  \, \beta^0_k +   e^* + e ,
         \label{model:lemma}
    \end{align}
    with ${\rm rank}(e^*) = R_{NT} \leq \min(N,T)/2$,
    $\| e \| = {\cal O}_P(\eta_{NT})$,
    $\|X_k\| = {\cal O}_P(\sqrt{NT})$,
    and
    $\frac 1 {\sqrt{NT}}  {\rm Tr}(X_k e^{\prime}) = {\cal O}_P(\xi_{NT})$,
    for $k=1,\ldots,K$.
    Then, the LS estimator 
    in \eqref{est:bai} calculated with $R = R_{NT}$
    factors in the estimation procedure,
    satisfies 
   $ \widehat \beta_{\rm LS} - \beta^0  =   {\cal O}_P \left( (\xi_{NT}+ R_{NT}\eta_{NT})  / \sqrt{NT} \right)$.
\end{lemma}

\begin{proof}[\bf  Proof of Lemma~\ref{lma:MWcons}]
   This proof is relatively minor modification
   of the consistency proof for the LS estimator 
   in \cite{MoonWeidner2015}, and more technical details 
   can be found there.
   For simplicity we just write $R$, $\eta$, $\xi$ instead of $R_{NT}$, $\eta_{NT}$, $\xi_{NT}$ in this proof. 
   We rewrite the definition of $\widehat \beta_{\rm LS} $ as
   \begin{align} 
      \widehat \beta_{\rm LS}  
      &= \argmin_{\beta}  {\cal L}_{NT}(\beta) ,
     \nonumber \\
      {\cal L}_{NT}(\beta)
      &:= 
      \min_{\left \{ \lambda \in \mathbbm{R}^{N\times R}, \;
                f \in \mathbbm{R}^{T\times R} \right\}  }
                \frac 1 {NT}
      {\rm Tr}\left[
      \left( Y - X \cdot \beta - \lambda f'
      \right)
      \left( Y - X \cdot \beta - \lambda f'
      \right)'
      \right] .
     \label{LSobjective} 
\end{align}
Since ${\rm rank}(e^*) = R$ 
we can write $e^* = \lambda^{*} f^{*\prime}$
for some $N \times R$ matrix    $\lambda^{*}$
and $T \times R$ matrix $f^{*\prime}$.

We now first establish a lower bound on ${\cal L}_{NT}(\beta)$.
   Let $\Delta \beta = \beta - \beta^0$.
   Consider the definition of
   ${\cal L}_{NT}(\beta)$ in equation \eqref{LSobjective} and plug in the model
    $Y=  \beta \cdot X + \lambda^{*} f^{*\prime} + e$.
    We then have
    \begin{align}
       {\cal L}_{NT}(\beta) &=  \min_{\left \{ \lambda \in \mathbbm{R}^{N\times R}, \;
                f \in \mathbbm{R}^{T\times R} \right\}  }   \frac 1 {NT}  {\rm Tr}
            \left[ \left( \Delta \beta \cdot X + e + \lambda^{*} f^{*\prime} - \lambda f'  \right)  
                   \left( \Delta \beta \cdot X + e + \lambda^{*} f^{*\prime} - \lambda f' \right)' \right] 
    \nonumber \\
        &\geq            \min_{\left \{ \tilde \lambda \in \mathbbm{R}^{N\times (2R)}, \;
                \tilde f \in \mathbbm{R}^{T\times (2R)} \right\}  }  \frac 1 {NT}  {\rm Tr}
            \left[ \left( \Delta \beta \cdot X + e   - \tilde \lambda \tilde f'  \right)  
                   \left( \Delta \beta \cdot X + e   - \tilde \lambda \tilde f' \right)' \right]    
    \nonumber \\
         &=
        \frac 1 {NT}   \min_{ \tilde f \in \mathbbm{R}^{T\times (2R)} }
         {\rm Tr}
            \left[ \left( \Delta \beta \cdot X + e \right) M_{\tilde f}
                   \left( \Delta \beta \cdot X + e \right)' \right]  
    \nonumber \\
         &=
   \frac 1 {NT} \min_{ \tilde f \in \mathbbm{R}^{T\times (2R)} }  \Bigg\{ 
         {\rm Tr}
            \left[ \left( \Delta \beta \cdot X  \right) M_{\tilde f}
                   \left(\Delta  \beta \cdot X   \right)' \right]   +  {\rm Tr} \left( e  e' \right)
                   -  {\rm Tr} \left( e P_{\tilde f}  e' \right)
        \nonumber \\ & \qquad \qquad        \qquad   \qquad \qquad        \qquad      \qquad        \qquad    
           + 2       {\rm Tr}
            \left[ \left( \Delta \beta \cdot X  \right)  e' \right]      
           -  2   {\rm Tr}
            \left[ \left( \Delta \beta \cdot X  \right) P_{\tilde f}  e' \right]   
           \Bigg\}
    \nonumber \\
         &\geq
   \frac 1 {NT}  \Bigg\{  
     \sum_{r=2R+1}^T  \mu_r\left[  (\Delta \beta \cdot X)' (\Delta \beta \cdot X)  \right]   
     +  {\rm Tr} \left( e  e' \right) -  2 R   \|e\|^2
        \nonumber \\ & \qquad \qquad        \qquad  \qquad        \qquad        \qquad        \qquad  
           + 2       {\rm Tr}
            \left[ \left( \Delta \beta \cdot X  \right)  e' \right]      
           -  4 R   \|e\|   \|    \Delta  \beta \cdot X \| 
           \Bigg\}
         \nonumber \\
           &\geq b \, \| \Delta \beta \|^2 + \frac 1 {NT} \, {\rm Tr} \left( e  e' \right)
           +  {\cal O}_P \left( \frac {R \, \eta^2} {NT} \right)
                + {\cal O}_P\left( \frac {(\xi + R  \,\eta) \, \| \Delta \beta \|}
                             {\sqrt{NT}} \right)     .
       \label{ConBnd1}
    \end{align}
   Here,  we applied the inequality $\left| {\rm Tr}(A) \right| \leq {\rm rank}(A) \|A\|$ with $A 
   = \left( \Delta \beta \cdot X \right)  P_{\tilde f} e'$
   and also with $A = e P_{\tilde f}  e' $.
   We also used that 
   $ \min_{ \tilde f }  
    {\rm Tr}
            \left[ \left( \Delta \beta \cdot X  \right) M_{\tilde f}
                   \left(\Delta  \beta \cdot X   \right)' \right]
           =  \sum_{r=2R+1}^T  \mu_r\left[  (\Delta \beta \cdot X)' (\Delta \beta \cdot X)  \right]   $.
 In the last step of \eqref{ConBnd1} we applied the 
 various assumptions in the lemma.
     
    Next, we establish an upper bound on ${\cal L}_{NT}(\beta^0)$. 
    We can choose $\lambda=\lambda^*$
    and $f = f^{*}$ in the minimization problem
    in   \eqref{LSobjective}, and therefore
    \begin{align}
       {\cal L}_{NT}(\beta^0)
          &\leq  \frac 1 {NT} \, {\rm Tr} \left( e  e' \right) .
       \label{ConBnd2}
    \end{align}
    Since we could choose $\beta=\beta^0$ in the minimization of $\beta$, the optimal
    $\widehat \beta_{\rm LS}$  needs to satisfy ${\cal L}_{NT}(\widehat \beta_{\rm LS} ) \leq {\cal L}_{NT}(\beta^0)$. 
    Together with \eqref{ConBnd1} and \eqref{ConBnd2} this gives
    \begin{align}
         b \, \| \widehat \beta_{\rm LS}  - \beta^0 \|^2
         &\leq
                            {\cal O}_P\left( \frac {(\xi+ R \, \eta) \| \widehat \beta_{\rm LS}  - \beta^0 \|} {\sqrt{NT}} \right)
                          +  {\cal O}_P \left( \frac {R \, \eta^2} {NT} \right)
                    \label{BoundBetaLSinequality}      
    \end{align}
    Since $R \rightarrow \infty$ as $N,T \rightarrow \infty$,
    we have 
    $$
    {\cal O}_P \left( \frac {R \, \eta^2} {NT} \right)
    \leq  {\cal O}_P \left( 
                          \left(  \frac {R \, \eta} {\sqrt{NT}}
                          \right)^2
                          \right)
                        \leq 
                         {\cal O}_P \left( 
                          \left(  \frac {\xi + R \, \eta} {\sqrt{NT}}
                          \right)^2
                          \right) ,
    $$
    and \eqref{BoundBetaLSinequality} thus implies
    \begin{align*}
           \| \widehat \beta_{\rm LS}  - \beta^0 \|^2
                          &\leq  {\cal O}_P\left( \frac {(\xi+ R \, \eta) \| \widehat \beta_{\rm LS}  - \beta^0 \|} {b \, \sqrt{NT}} \right)
                          + {\cal O}_P \left(  \frac 1 b
                          \left(  \frac {\xi + R \, \eta} { \sqrt{NT}}
                          \right)^2
                          \right)
                \\          
                    &=:  2 B_1 \,       \| \widehat \beta_{\rm LS}  - \beta^0 \|
                    + (B_2)^2 ,
    \end{align*}    
    with random variables $B_1 = {\cal O}_P \left( 
                            \frac {\xi + R \, \eta} {\sqrt{NT}}
                           \right)$
                        and $B_2 = {\cal O}_P \left( 
                            \frac {\xi + R \, \eta} {\sqrt{NT}}
                           \right)$,
                           and where we used that
                           $b$ is a positive constant.                           
   Completing the square gives
     \begin{align*}  
          \left( \| \widehat \beta_{\rm LS}  - \beta^0 \| -  B_1  \right)^2
          \leq (B_2)^2 +  (B_1)^2  ,
      \end{align*} 
     by taking the square root we thus obtain   
     \begin{align*}  
           \| \widehat \beta_{\rm LS}  - \beta^0 \| 
          \leq  B_1  + \sqrt{ (B_2)^2  +  (B_1)^2 }  .
      \end{align*} 
    Since    $B_1$ and $B_2$  
    are both of order ${\cal O}_P \left( 
                            \frac {\xi + R \, \eta} {\sqrt{NT}}
                           \right)$
   it thus follows that
    $$\| \widehat \beta_{\rm LS}  - \beta^0 \|  =
         {\cal O}_P \left( (\xi+ R \, \eta)  / \sqrt{NT} \right),$$
    which is what we wanted to show.
\end{proof}

Using Lemma~\ref{lma:MWcons} we are now ready to 
prove Theorem \ref{thm:bai.con}.

\begin{proof}[\bf{Proof of Theorem \ref{thm:bai.con}}]
To apply Lemma~\ref{lma:MWcons} we first need to
define $e$ and $e^*$ such that
\eqref{model:lemma} is an implication
of our model \eqref{eq:consistency.model:general}.
Decompose $\Gamma = \sum_{r = 1}^{\min\{N,T\}}\lambda^*_r f^{*\prime}_r$, which is a reformulation of the singular value decomposition of a matrix. Define $e^* =  \sum_{r = 1}^{R_{NT}} \lambda^*_r f^{*\prime}_r$ such that ${\rm rank}(e^*) = R_{NT}$. Also define $e =S + \varepsilon$ where $S = \Gamma -  \sum_{r = 1}^{R_{NT}} \lambda^*_r f^{*\prime}_r$. With these definitions model \eqref{eq:consistency.model:general} can be rewritten as \eqref{model:lemma} and it remains to show Assumptions \ref{ass:SN}-\ref{ass:NC} are sufficient for Lemma~\ref{lma:MWcons}  and to characterise the sequences $\eta_{NT}$ and $\xi_{NT}$.

First, use the norm inequality $\norm{S+\varepsilon}\leq\norm{S} + \norm{\varepsilon}$ with $\norm{\varepsilon} = O_P(\sqrt{\max\{N,T\}})$ from Assumption \ref{ass:SN} (ii) to show $\norm{e} \leq \norm{S} + O_P(\sqrt{\max\{N,T\}})$. To bound $ \norm{S}$ use the fact that the spectral norm is bounded by the Frobenius norm and Assumption~\ref{ass:SVD} to show 
\begin{align*}
     \norm{S}^2\leq\norm{S}_F^2 &= \sum_{r = R_{NT} + 1}^\infty \sigma_r^2(\Gamma) \quad \quad &\quad \quad {\rm }\\
    &\leq O_P\big({NT}R_{NT}^{1-2\rho}\big).
\end{align*}
This shows that $\norm{e}$ is asymptotically bounded in probability by the sequence $\eta_{NT}$ with 
\begin{align*}
    \eta_{NT} = \sqrt{\max\{N,T\}} + \sqrt{NT}R_{NT}^{(1-2\rho)/2}.
\end{align*}
That is, $\norm{e} = O_P(\eta_{NT})$. 

Secondly, the bound on $\norm{X_k}$ is direct from Assumption ~\ref{ass:SN}.(i) again because the spectral norm is bounded by the Frobenius norm. That is, $\norm{X_k}^2 \leq \norm{X_k}_F^2 = \sum_{i=1}^N\sum_{t=1}^T X_{it,k}^2 = O_P(NT)$. 

Lastly, we need to show that $\frac{1}{\sqrt{NT}}{\rm Tr}(X_k e^{\prime}) = O_P(\xi_{NT})$ and to find $\xi_{NT}$. To do this we decompose $e$ and use the Cauchy-Schwarz inequality, the triangle inequality and linearity of the trace operator in the following, 
 \begin{align}\label{proof1:WE}
 \begin{split}
     \left|\frac{1}{\sqrt{NT}}\textrm{Tr}(X_k e^\prime)\right| &=  \left|\frac{1}{\sqrt{NT}}\textrm{Tr}(X_k (S + \varepsilon)^\prime)\right| \\
     &\leq \frac{1}{\sqrt{NT}}\norm{X_k}_F\norm{S}_F + \frac{1}{\sqrt{NT}}|\textrm{Tr}(X_k \varepsilon^\prime)|\\
     &=O_P(1)\norm{S}_F + O_P(1).
     \end{split}
 \end{align}
The third line follows from Assumption \ref{ass:SN}.(i) and Assumption \ref{ass:EX}. From above we know $\norm{S}_F = O_P\big(\sqrt{NT}R_{NT}^{(1-2\rho)/2}\big)$, hence we have found $\xi_{NT} = \sqrt{NT}R_{NT}^{(1-2\rho)/2} + 1$.

Thus, we have shown that all conditions for Lemma~\ref{lma:MWcons} are satisfied and found the rates $\eta_{NT}$ and $\xi_{NT}$. This shows that LS estimation in \eqref{est:bai} on the model \eqref{eq:consistency.model:general} with $R=R_{NT}$ factors satisfies $\widehat \beta_{\rm LS} - \beta^0  =   O_P \left( (\xi_{NT}+ R_{NT}\eta_{NT})  / \sqrt{NT} \right)$, with 
\begin{align*}
     O_P\Bigg(\frac{(\xi_{NT}+ R_{NT}\eta_{NT})}{\sqrt{NT}}\Bigg)  &= {O}_P\Big(R_{NT}^{(1-2\rho)/2}\Big) + 
     {O}_P\Big(\frac{1}{\sqrt{NT}}\Big) + 
     {O}_P\Big(R_{NT}^{(3-2\rho)/2}\Big) \\
      &\qquad + \mathcal{O}_P\Bigg(R_{NT}\sqrt{\frac{\max\{N,T\}} {NT}}\Bigg) \\
     &= {O}_P\Big(R_{NT}^{(3-2\rho)/2}\Big) + 
     {O}_P\Big(R_{NT}\min\{N,T\}^{-1/2}\Big).
\end{align*}
\end{proof}

\begin{proof}[\bf{Proof of Remark 1}]
Note that if we weaken the singular value decay to that supposed in Remark 1, i.e. $\sigma_r(\Gamma) = c\sqrt{NT}r^{-\rho}$, and otherwise maintain Assumptions~\ref{ass:SN}-\ref{ass:NC} we can further bound the bias in LS estimation found in Theorem~\ref{thm:bai.con} as follows. For $\norm{S}_F$, note,
    \begin{align*}
    \norm{S}_F^2 &= \sum_{r = R_{NT} + 1}^\infty \sigma_r^2(\Gamma) \quad \quad &\quad \quad {\rm }\\
    &\leq \sum_{r = R_{NT} + 1}^\infty c{NT} r^{-2\rho} \quad \quad wpa.1& \quad \quad {\rm (Assumption \ \ref{ass:SVD})}\\
    &\leq c{NT} \int_{R_{NT}}^\infty r^{-2\rho}dr \quad \quad wpa.1&\quad \quad {\rm (integral \ bound)}\\
    &= \frac c {2\rho - 1} {NT}R_{NT}^{1-2\rho} \quad \quad wpa.1&
\end{align*}
In the third line we use an integral bound and the fourth line simply evaluates this integral. From line two all arguments are $wpa.1$, hence $\norm{S}_F = \mathcal{O}_P(\sqrt{NT}R_{NT}^{(1-2\rho)/2})$, where $(c/2\rho - 1)$ is the bounding constant. We can then directly bound
\begin{align*}
    \norm{S} &= \max_{r\in\{R_{NT}+1,\dots, \min\{N,T\}\}}\sigma_r(\Gamma)\\
    &= \mathcal{O}_P\big(\sqrt{NT}(R_{NT}+1)^{-\rho}\big),
\end{align*}
where we use the convention that singular values are indexed in descending order. We then simplify the last bound to $\norm{S}= \mathcal{O}_P\big(\sqrt{NT}R_{NT}^{-\rho}\big)$, replacing $R_{NT}+1$ with $R_{NT}$ as $R_{NT}\rightarrow\infty$. We can then rely on the same working in the proof of Theorem~\ref{thm:bai.con} to show that the conditions in Lemma~\ref{lma:MWcons} are satisfied with 
$\xi_{NT} = \sqrt{NT}R_{NT}^{(1-2\rho)/2} + 1$
and 
$\eta_{NT} = \sqrt{\max\{N,T\}} + \sqrt{NT}R_{NT}^{-\rho}$, 
where the second term in $\eta_{NT}$ is slightly different to Theorem~\ref{thm:bai.con}.
Hence,  $\widehat \beta_{\rm LS} - \beta^0  =   O_P \left( (\xi_{NT}+ R_{NT}\eta_{NT})  / \sqrt{NT} \right)$, with
\begin{align*}
     O_P\Bigg(\frac{(\xi_{NT}+ R_{NT}\eta_{NT})}{\sqrt{NT}}\Bigg)  &= {O}_P\Big(R_{NT}^{(1-2\rho)/2}\Big) +
     {O}_P\Big(\frac{1}{\sqrt{NT}}\Big) + 
     {O}_P\Big(R_{NT}^{1-\rho}\Big)\\
      &+ {O}_P\Bigg(R_{NT}\sqrt{\frac{\max\{N,T\}} {NT}}\Bigg) \\
     &= {O}_P\Big(R_{NT}^{1-\rho}\Big) +
     {O}_P\Big(R_{NT}\min\{N,T\}^{-1/2}\Big).
\end{align*}
\end{proof}

To prove Lemma \ref{lemma:SG} we rely on the following result from \citep{griebel2014approximation}, which we state without proof.

Let $ H^p(\Omega_\alpha \times \Omega_\gamma)$ denote the Sobolev space $W^{p,k}$ on the product domain $(\Omega_\alpha \times \Omega_\gamma)$ for $k = 2$, which is in turn a Hilbert space. 
In the one dimensional case, this space admits functions in $L^2(\mathbb{R})$-space whose derivatives up to order $p$ are also in $L^2(\mathbb{R})$-space. 
In multiple dimensions this definition extends as follows. 
Let $\nabla := \{\nabla^\alpha, \nabla^\gamma\}$ be a multi-index that captures all the dimensions of $\alpha$ and $\gamma$ respectively. 
Define the mixed partial derivative as,
\begin{align*}
    f^{(\nabla)} = \frac{\partial^{|\nabla|}f}{\partial a_1^{\nabla^\alpha_1}\dots \partial a_{d_\alpha}^{\nabla^\alpha_{d_\alpha}}  \partial c_1^{\nabla^\gamma_1}\dots \partial c_{d_\gamma}^{\nabla^\gamma_{d_\gamma}}},
\end{align*}
where $a\in\Omega_\alpha$ and $c\in\Omega_\gamma$ with $(\Omega_\alpha \times \Omega_\gamma)$ the domain of $f$. 
Then $|\nabla| = |\nabla_\alpha| + |\nabla_\gamma|$ and the bivariate function $\fct$ is said to be in Hilbert space of order $p$ if the mixed partial derivative exists (weakly) and 
\begin{align*}
    \norm{\fct^{(\nabla)}}_{L_2} \leq \infty \textrm{ for all } |\nabla| \leq p.
\end{align*}
This space of functions places a bound on the function itself as well as its derivative, which is why we refer to it as a smoothness condition.

\begin{lemma}[Theorem 3.5 in \citealt{griebel2014approximation}]
\label{thm3.4}
Let $\fct \in H^p(\Omega_\alpha \times \Omega_\gamma)$  and $p>\min\left\lbrace n_\alpha , n_\gamma \right\rbrace/2$, then 
\begin{align}
\norm{\fct - \sum_{l=1}^R{\sigma_l}(\varphi_l\otimes \psi_l)}_{L^2(\Omega_1 \times \Omega_2)} =  {O}\left(R^{\frac{1}{2} - \frac{p}{\min\left\lbrace n_\alpha , n_\gamma \right\rbrace}}\right).
\end{align}
\end{lemma}

In the following proof we use the Frobenius norm, which as a reminder is defined as $\norm{A}_F^2 = \sum_{i=1}^N\sum_{t=1}^T |A_{it}|^2$ for any $N\times T$ matrix $A$. 
\begin{proof}[\bf{Proof of Lemma \ref{lemma:SG}}]
From Lemma \ref{thm3.4} we have, 
 \begin{align}
    \begin{split}%
        \mathbb{E}\Big[\Big(\fct(\alpha_i,\gamma_t) - \sum_{s =1}^R \sigma_r \varphi_r(\alpha_i) \psi_r(\gamma_t)\Big)^2\Big] &= \int_{\Omega_\alpha} \int_{\Omega_\gamma} \big(\fct(a,c) - \sum_{s =1}^R \sigma_r \varphi_r(a) \psi_r(c) \big)^2 f_{\alpha_i,\gamma_t}(a,c)da dc
        \\
        &\leq \int_{\Omega_\alpha} \int_{\Omega_\gamma} \big(\fct(a,c) - \sum_{s =1}^R \sigma_r \varphi_r(a) \psi_r(c) \big)^2dadc  \sup_{a,c}f_{\alpha_i,\gamma_t}(a,c)
        \\
        &= \norm{\fct - \sum_{l =1}^R \sigma_l \varphi_l \otimes \psi_l}_{L^2(\Omega_\alpha \times \Omega_\gamma)}^2{O}(1) 
        \\
        &={O}\left(R^{{1} - \frac{2p}{\min\left\lbrace n_\alpha , n_\gamma \right\rbrace}}\right) ,
    \end{split}
\end{align}
where in the second line we use a supremum bound on the probabilities, in the third line we use the definition of the ${L^2(\Omega_\alpha \times \Omega_\gamma)}$-norm and in the final line we use Lemma~\ref{thm3.4}. This shows that, in expectations, the entry-wise functional representation decays at polynomial rate $r^{1-2\rho}$, with $\rho = p/\min\left\lbrace n_\alpha , n_\gamma \right\rbrace$.  

Using the Markov inequality gives
\begin{align*}
    \Big(\Gamma_{it} - \sum_{\ell= 1}^{r} \sigma_\ell \varphi_\ell(\alpha_i) \psi_\ell(\gamma_t)^\prime\Big)^2 = \mathcal{O}_P\left(r^{{1} - \frac{2p}{\min\left\lbrace n_\alpha , n_\gamma \right\rbrace}}\right),
\end{align*}
which we use to bound singular values of the matrix $\Gamma$ as follows. 

We know
\begin{align*}
   \Gamma_{it} =\fct(\alpha_i, \gamma_t) &= \sum_{r=1}^{\infty}\sigma_r \varphi_r(\alpha_i) \psi_r(\gamma_t)
   = \sum_{r=1}^{\infty}\sigma_r w_{ir} v_{tr}
\end{align*}
and in matrix form, 
\begin{align*}
   \Gamma =\fct(\alpha, \gamma) &= \sum_{r=1}^{\infty}\sigma_r \varphi_r(\alpha) \psi_r(\gamma)^\prime
   = \sum_{r=1}^{\infty}\sigma_r w_{r} v_{r}^\prime.
\end{align*}
Hence, we have
\begin{align*}
    \sum_{\ell = r + 1}^{\min\{N,T\}}\sigma_\ell^2(\Gamma) &= \min_{\lambda \in \mathbb{R}^{N \times r}}
   \min_{f \in \mathbb{R}^{T \times r}}
  \left\| \Gamma - \lambda \, f' \right\|_F^2
  \\
  &\leq  \norm{\Gamma - \sum_{\ell= 1}^{r} \sigma_\ell \varphi_\ell(\alpha) \psi_\ell(\gamma)^\prime}_F^2
  \\ 
  &= \sum_i \sum_t \Bigg(\sum_{\ell= r + 1}^{\infty}\sigma_\ell \varphi_\ell(\alpha_i) \psi_\ell(\gamma_t)\Bigg)^2
  \\
    &= \sum_i \sum_t \mathcal{O}_P\left(r^{{1} - \frac{2p}{\min\left\lbrace n_\alpha , n_\gamma \right\rbrace}}\right)
    \\
    &= NT \mathcal{O}_P\left(r^{{1} - \frac{2p}{\min\left\lbrace n_\alpha , n_\gamma \right\rbrace}}\right).
\end{align*}
Hence, we have $\frac{1}{NT}\sum_{\ell = r+1}^{\min\{N,T\}}\sigma_\ell^2(\Gamma) =  \mathcal{O}_P\big(r^{1- 2\rho}\big)$ with $\rho = p/\min\left\lbrace n_\alpha , n_\gamma \right\rbrace$ , and Assumption~\ref{ass:SVD} is satisfied. 
\end{proof}

\subsection{Proofs for Section~\ref{sec:GroupedFE}}

\begin{proof}[\bf{Proof of Lemma~\ref{lemma:ApproxError}}]

From Section~\ref{sec:GroupedFE} we have 
\begin{align*}
  \kappa_{NT} &:= \left( \sum_{i = 1}^N \sum_{t = 1}^T \widetilde{X}_{it}^\prime \widetilde{X}_{it}\right)^{-1} \sum_{i = 1}^N \sum_{t = 1}^T \widetilde{X}_{it}^\prime \, \widetilde \Gamma_{it},
\end{align*}
with $\widetilde \Gamma$ defined analogously to $\widetilde{X}_k$ and $\widetilde{Y}$.

Take
\begin{align*}
  \norm{\kappa_{NT} }&:= \norm{\left( \sum_{i = 1}^N \sum_{t = 1}^T \widetilde{X}_{it}^\prime \widetilde{X}_{it}\right)^{-1} \sum_{i = 1}^N \sum_{t = 1}^T \widetilde{X}_{it}^\prime \, \widetilde \Gamma_{it}}.
\end{align*}
Using the inequality ${\norm{Az}} \leq\norm{A} {\norm{z}}$ for general matrices $A$ and vectors $z$ we find
\begin{align*}
    \left\|  \kappa_{NT} \right\|
    \leq  \left\| \left(\sum_{i = 1}^N \sum_{t = 1}^T \widetilde{X}_{it}^\prime \widetilde{X}_{it}\right)^{-1} \right\|
        \, \norm{\sum_{i = 1}^N \sum_{t = 1}^T\widetilde{X}_{it}^\prime  
         \widetilde \Gamma_{it} }.
\end{align*}

Use $\left|\sum_{i = 1}^N \sum_{t = 1}^T\widetilde{X}_{it,k}\widetilde \Gamma_{it}\right| \leq \sum_{i = 1}^N \sum_{t = 1}^T\left| \widetilde{X}_{it,k}\widetilde \Gamma_{it}\right|$ and H\"older's inequality such that 
\begin{align*}
    \begin{bmatrix}
       \left|\sum_{i = 1}^N \sum_{t = 1}^T\widetilde{X}_{it,1}  
     \widetilde \Gamma_{it}\right|  \\
       \vdots \\
      \left|\sum_{i = 1}^N \sum_{t = 1}^T\widetilde{X}_{it,K}  
     \widetilde \Gamma_{it}\right|
     \end{bmatrix}  \leq \begin{bmatrix}
       \sum_{i = 1}^N \sum_{t = 1}^T\left| \widetilde{X}_{it,1}\widetilde \Gamma_{it}\right|  \\
       \vdots \\
       \sum_{i = 1}^N \sum_{t = 1}^T\left| \widetilde{X}_{it,K}\widetilde \Gamma_{it}\right|
     \end{bmatrix} \leq   \begin{bmatrix}
       \norm{\vect(X_1)}_\infty \\
       \vdots \\
       \norm{\vect(X_K)}_\infty
     \end{bmatrix}\norm{\vect (\widetilde \Gamma)}_1, 
\end{align*}
{where $\vect(A)$ vectorises a matrix $A$ such that $\norm{\vect(A)}_\infty = \max_{i,t}|A_{it}|$ yields the maximum norm and $\norm{\vect(A)}_1 = \sum_{i=1}^N\sum_{t=1}^T |A_{it}|$ yields the entry-wise 1-norm of such a matrix. }

Take the $\norm{\cdot}$ to show 
\begin{align*}
    \norm{\sum_{i = 1}^N \sum_{t = 1}^T\widetilde{X}_{it}^\prime  \widetilde \Gamma_{it} } &= \left( \sum_k \bigg| \sum_{i = 1}^N \sum_{t = 1}^T \widetilde{X}_{it,k} \widetilde \Gamma_{it}\bigg|^2 \right)^{1/2}\\
    &\leq \left( \sum_k \left(\big\|\vect( \widetilde{X}_k)\big\|_\infty \big\|\vect (\widetilde \Gamma)\big\|_1 \right)^2 \right)^{1/2}\\
    &= \left( \sum_k \left(\big\|\vect( \widetilde{X}_k)\big\|_\infty \right)^2 \right)^{1/2}\big\| \vect (\widetilde \Gamma)\big\|_1 \leq \left(\sum_k \big\|\vect( \widetilde{X}_k)\big\|_\infty \right)\big\|\vect (\widetilde \Gamma)\big\|_1,
\end{align*}
where in the last line we use that $\norm{\vect (\widetilde \Gamma)}_1$ is a scalar and that $\big\|\vect (X_k)\big\|_\infty>0 \,\,\, \forall\,k$.
Thus we can bound the norm of $\kappa_{NT}$ by 
\begin{align*}
    \left\|  \kappa_{NT} \right\|
    \leq  \left\| \left(\sum_{i = 1}^N \sum_{t = 1}^T \widetilde{X}_{it}^\prime \widetilde{X}_{it}\right)^{-1} \right\|
        \, \left( \sum_{k=1}^K \big\|\vect( \widetilde{X}_k) \big\|_\infty  \right)
        \big\|\vect (\widetilde \Gamma) \big\|_1.
\end{align*}

Concentrate on $\big\| \vect (\widetilde \Gamma)  \big\|_1$. Let $n_i^N$ be the size of each $i$'s cluster and $n_t^T$ be the size of each $t$'s cluster, then 
\begin{align*}
    {\widetilde \Gamma_{it} }&= \fct(\alpha_i, \gamma_t) - \frac{1}{n_i^N}\sum_{j \in g_i}\fct(\alpha_j, \gamma_t) - \frac{1}{n_t^T}\sum_{s \in c_t}\fct(\alpha_i, \gamma_s) + \frac{1}{n_i^N}\frac{1}{n_t^T}\sum_{j \in g_i}\sum_{s \in c_t}\fct(\alpha_j, \gamma_s).
\end{align*}
Take the following Taylor expansions,
\begin{align*}
    \fct(\alpha_j,\gamma_s) &= \fct( \alpha_{i}, \gamma_{t}) + \frac{\partial \fct( \alpha_i, \gamma_{t})} {\partial \alpha'} 
  (\alpha_j -  \alpha_{i}) + \frac{\partial \fct( \alpha_{i}, \gamma_{t})} {\partial \gamma'} 
  (\gamma_s -  \gamma_{t})
  + r(i,j,t,s) 
  \\ 
  \fct({\alpha}_j,{\gamma}_{t}) &= \fct( \alpha_{i}, \gamma_{t}) + \frac{\partial \fct( \alpha_{i},  \gamma_{t})} {\partial \alpha'_i} 
  (\alpha_j -  \alpha_{i}) + r^\prime(i, j, t)\\ 
  \fct( \alpha_{i}, \gamma_{s}) &= \fct( \alpha_{i}, \gamma_{t}) +  \frac{\partial \fct( \alpha_{i},  \gamma_{t})} {\partial \gamma'} 
  (\gamma_s -  \gamma_{t}) + r^{\prime\prime}( t,s,i),
\end{align*}
where $r$, $r^\prime$ and $r^{\prime\prime}$ are remainder terms from the Taylor expansion. 

From these expansions we have 
\begin{align*}
     \frac{1}{n_i^N}\sum_{j \in g_i}\fct(\alpha_j,\gamma_t) = \fct(\alpha_i,\gamma_t) + \frac{1}{n_i^N}\sum_{\substack{j \in g_i,\\ j\neq i}}\left(\frac{\partial \fct( \alpha_{i},  \gamma_{t})} {\partial \alpha'} (\alpha_j -  \alpha_{i}) + r^\prime(i, j, t)\right),
\end{align*}
\begin{align*}
     \frac{1}{n_t^T}\sum_{s \in c_t}\fct(\alpha_i,\gamma_s) = \fct(\alpha_i,\gamma_t) + \frac{1}{n_t^T}\sum_{\substack{s \in c_t,\\ s\neq t}}\left(\frac{\partial \fct( \alpha_{i},  \gamma_{t})} {\partial \gamma'} (\gamma_s -  \gamma_t) + r^{\prime\prime}(t, s, i)\right),
\end{align*}
and 
\begin{align*}
    \frac{1}{n_i^N}\frac{1}{n_t^T}&\sum_{j \in g_i}\sum_{s \in c_t}\fct(\alpha_j, \gamma_s) = \frac{1}{n_i^Nn_t^T}\fct(\alpha_i,\gamma_t) + \frac{1}{n_i^Nn_t^T}\left( \sum_{\substack{j \in g_i,\\ j\neq i}}\sum_{\substack{s \in c_t,\\ s\neq t}} \fct(\alpha_j,\gamma_s) + \sum_{\substack{j \in g_i,\\ j\neq i}}\fct(\alpha_j,\gamma_t) + \sum_{\substack{s \in c_t,\\ s\neq t}}\fct(\alpha_i,\gamma_s)  \right)\\
    &= \fct(\alpha_i,\gamma_t) + \frac{1}{n_i^Nn_t^T}\sum_{\substack{j \in g_i,\\ j\neq i}}\sum_{\substack{s \in c_t,\\ s\neq t}}\left(\frac{\partial \fct( \alpha_{i},  \gamma_{t})} {\partial \alpha'}(\alpha_j -  \alpha_{i}) + \frac{\partial \fct( \alpha_{i},  \gamma_{t})} {\partial \gamma'} (\gamma_s -  \gamma_t) + r(i,j,t,s) \right)\\
    &+\frac{1}{n_i^Nn_t^T}\sum_{\substack{j \in g_i,\\ j\neq i}}\left(\frac{\partial \fct( \alpha_{i},  \gamma_{t})} {\partial \alpha'}(\alpha_j -  \alpha_{i}) + r^\prime(i,j,t) \right) + \frac{1}{n_i^Nn_t^T}\sum_{\substack{s \in c_t,\\ s\neq t}}\left(\frac{\partial \fct( \alpha_{i},  \gamma_{t})} {\partial \gamma'} (\gamma_s -  \gamma_t) + r^{\prime\prime}(t,s,i) \right)\\
    &= \fct(\alpha_i,\gamma_t) + \frac{1}{n_i^N}\sum_{\substack{j \in g_i,\\ j\neq i}}\left(\frac{\partial \fct( \alpha_{i},  \gamma_{t})} {\partial \alpha'}(\alpha_j -  \alpha_{i}) + r^\prime(i,j,t) \right) + \frac{1}{n_t^T}\sum_{\substack{s \in c_t,\\ s\neq t}}\left(\frac{\partial \fct( \alpha_{i},  \gamma_{t})} {\partial \gamma'} (\gamma_s -  \gamma_t) + r^{\prime\prime}(t,s,i)\right)\\
    &+ \frac{1}{n_i^Nn_t^T}\sum_{\substack{j \in g_i,\\ j\neq i}}\sum_{\substack{s \in c_t,\\ s\neq t}}r(i,j,t,s).
\end{align*}
We explicitly split the sum in the second line to make clearer the fact that almost all terms cancel out once we difference these identities. From the last line it should be clear that,
\begin{align*}
    {\widetilde \Gamma_{it} }= \frac{1}{n_i^Nn_t^T}\sum_{\substack{j \in g_i,\\ j\neq i}}\sum_{\substack{s \in c_t,\\ s\neq t}}r(i,j,t,s).
\end{align*}
From $\fct(.,.)$ being twice continuously differentiable and a uniformly bounded second derivative, we have  from Cauchy-Schwarz $r(i,j,t,s)= 
O\left(
\norm{\alpha_i -  \alpha_j}^2 + 
\norm{\gamma_t -  \gamma_s}^2 \right)$. 

For the entry-wise 1-norm, we have,  
\begin{align*}
    \norm{\vect (\widetilde \Gamma)}_1 &= \sum_i\sum_t\left|\frac{1}{n_i^Nn_t^T}\sum_{\substack{j \in g_i,\\ j\neq i}}\sum_{\substack{s \in c_t,\\ s\neq t}} r(i,j,t,s) \right|\\
    &\leq 
    \sum_i\sum_t\left|\frac{1}{n_i^Nn_t^T}\sum_{\substack{j \in g_i,\\ j\neq i}}\sum_{\substack{s \in c_t,\\ s\neq t}}O\left(\norm{\alpha_i -  \alpha_j}^2 \right) \right|
    + \sum_i\sum_t\left|\frac{1}{n_i^Nn_t^T}\sum_{\substack{j \in g_i,\\ j\neq i}}\sum_{\substack{s \in c_t,\\ s\neq t}}O\left(\norm{\gamma_t -  \gamma_s}^2 \right) \right|.
\end{align*}

Now, concentrate on the  first term, 
\begin{align*}
    \sum_i\sum_t\left|\frac{1}{n_i^Nn_t^T}\sum_{\substack{j \in g_i,\\ j\neq i}}\sum_{\substack{s \in c_t,\\ s\neq t}}O\left(\norm{\alpha_i -  \alpha_j}^2 \right) \right| &\leq \sum_i\sum_t\left|\frac{(n_i^N - 1)(n_t^T - 1)}{n_i^Nn_t^T}\max_{\substack{j \in g_i,\\ j\neq i}}O\left(\norm{\alpha_i -  \alpha_j}^2 \right) \right|\\
    &= O(T)\sum_i\max_{\substack{j \in g_i,\\ j\neq i}}\norm{\alpha_i -  \alpha_j}^2
\end{align*}
Use Assumption~\ref{ass:ApproxError}(iii) to show for $j\in g_i$,
\begin{align*}
    \norm{\alpha_i - \alpha_j}^2 &\leq B^2   \norm{\lambda(\alpha_i) - \lambda(\alpha_j) }^2\\
    &=  B^2   \norm{\lambda(\alpha_i) - \widehat\lambda_i - ( \lambda(\alpha_j)- \widehat\lambda_j) + \widehat\lambda_i - \widehat\lambda_j}^2\\
    &\leq B^2  \left(\norm{\lambda(\alpha_i) - \widehat\lambda_i} +   \norm{ \lambda(\alpha_j)- \widehat\lambda_j} +  \norm{\widehat\lambda_i - \widehat\lambda_j}\right)^2.
\end{align*}
An application of Cauchy-Schwarz and  Assumption~\ref{ass:ApproxError}(iv) gives
 \begin{align*}
    \sum_{i = 1}^N\max_{\substack{j \in g_i,\\ j\neq i}}\norm{\alpha_i - \alpha_j}^2 
    &\leq  
    B^2 \sum_{i = 1}^N\max_{\substack{j \in g_i,\\ j\neq i}} \bigg( \norm{\lambda(\alpha_i) - \widehat\lambda_i}^2 +   \norm{ \lambda(\alpha_j)- \widehat\lambda_j}^2 +  \norm{\widehat\lambda_i - \widehat\lambda_j}^2 
    \bigg),
\end{align*}
hence we have 
\begin{align*}
    \sum_{i= 1}^N\sum_{t=1}^T\left|\frac{1}{n_i^Nn_t^T}\sum_{\substack{j \in g_i,\\ j\neq i}}\sum_{\substack{s \in c_t,\\ s\neq t}}O\left(\norm{\alpha_i -  \alpha_j}^2 \right) \right| = NTO_P(\xi_{NT}).
\end{align*}
The $t$-dimension analogy is direct such that  $ \sum_{i= 1}^N\sum_{t=1}^T\left|\frac{1}{n_i^Nn_t^T}\sum_{\substack{j \in g_i,\\ j\neq i}}\sum_{\substack{s \in c_t,\\ s\neq t}}O\left(\norm{\gamma_t -  \gamma_s}^2 \right) \right|= NTO_P(\xi_{NT})$.

Lastly, use Assumption~\ref{ass:ApproxError}.(vi), which implies $\left(\sum_{i = 1}^N \sum_{t = 1}^T \widetilde{X}_{it}^\prime \widetilde{X}_{it}\right)^{-1} = O_p(1/NT)$, to show \begin{align*}
     \left\|  \kappa_{NT} \right\| &= O_p(\xi_{NT})\\
     \Rightarrow  \kappa_{NT}  &= O_p(\xi_{NT})
\end{align*}

\end{proof}

For each partition $\mathcal{O}_q$, with $q\in\{1,2,3,4\}$, the $N_q\times G^{(q)}$ matrix $D_\nu^{(q)}$, respectively $T_q\times C^{(q)}$ matrix $D_\delta^{(q)}$, represent the $i$, respectively $t$, cluster assignment matrices for $(i,t)\in\mathcal{O}_q$ where the columns of each matrix are binary indicators of cluster assignment. 
That is, any given column of $D_\nu^{(q)}$ represents a cluster equal to 1 if that row is a member of the cluster and 0 otherwise, and likewise for $D_\delta^{(q)}$. Here $G^{(q)}$ are the number of $i$ clusters and $C^{(q)}$ are the number of $t$ clusters in $\mathcal{O}_q$. 
For each partition define the annihilation matrix $M_\nu^{(q)} = \mathbb{I}_{N_q} -  D_\nu^{(q)}\left( [D_\nu^{(q)}]^\prime D_\nu^{(q)} \right)^{-1}[D_\nu^{(q)}]^\prime$ and $M_\delta^{(q)} =  \mathbb{I}_{T_q} -  D_\delta^{(q)}\left( [D_\delta^{(q)}]^\prime D_\delta^{(q)} \right)^{-1}[D_\delta^{(q)}]^\prime$. 
To perform within-cluster mean-differences we can then take, for matrix $A^{(q)}$ being the partition $\mathcal{O}_q$ of matrix $A$, $\check{A}^{(q)} = M_\nu^{(q)}A^{(q)}M_\delta^{(q)}$.\footnote{Note these are very similar to the $\tilde{A}$ variables in the main text, but here we make the distinction that projection is done at the partition level.} 
Take $\check{A}$ as the block matrix with blocks $\check{A}^{(q)}$. Further, for each regressor, $k$, let $\check{X}_k$ be defined similarly for each $k$ separately such that $\check{X}_{it}$ a $K$ dimensional column vector. 

\begin{assumption}
    \label{ass:ApproxErrorSplit}
    Let $\mathcal{O}_q$ denote partitions for cluster formation and $\mathcal{O}_q^*$ denote partitions for proxy sampling. 
    Across each partition, $\alpha_i^{(q)}$ has common support $\mathcal{A}$ for each $q$, $\gamma^{(q)}$ has common support $\mathcal{C}$ for each $q$, and both of these are bounded and convex sets. 
    Also, assume each partition is of equal size, up to rounding error, such that they all grow proportionally with $N,T$. 
   There exists a sequence $\xi_{NT}>0$ common to all partitions such that $\xi_{NT}  \to 0$ as $N,T \to \infty$, and
    \begin{enumerate}[(i)]     

        \item The function $\fct(\cdot,\cdot)$ is  at least twice continuously differentiable with uniformly bounded second derivatives.
        
        \item For each $q$, every unit $i\in\mathcal{O}_q$ is a member of exactly one 
        group $g_i^{(q)} \in \{1,\ldots,G^{(q)}\}$, and every time period $t$ is 
        a member of exactly one group $c_t^{(q)} \in \{1,\ldots,C^{(q)}\}$. The size
        of all $G^{(q)}$ groups of units, and the size
        of all $C^{(q)}$ groups of time periods is 
        bounded uniformly by $Q_{\max}$ for all $q$.
        
        \item There exists $B>0$ such that for all $q$ there is,
         $\left\|a - b \right\| \leq B \left\| \lambda^{(q)}(a) -  \lambda^{(q)}(b) \right\|$ for all $a,b\in {\cal A} $
        , and  $\left\| a - b \right\| \leq B \left\| f^{(q)}(a) -  f^{(q)}(b) \right\|$ for all $a, b \in {\cal C}$.
              
       \item For each $q$ there is, \\
        $ \frac 1 {N^*_q T^*_q}  \sum_{i=1}^{N}\sum_{t=1}^{T} \mathbbm{1}\{(i,t)\in\mathcal{O}_{q}^*\} \left(   \left\| \widehat{\lambda}^{(q)}_i - \lambda^{(q)}(\alpha_i) \right\|^2 
        \right)= O_P \left( \xi_{NT} \right) $,  \\
        $ \frac 1 {N^*_q T^*_q}  \sum_{i=1}^{N}\sum_{t=1}^{T} \mathbbm{1}\{(i,t)\in\mathcal{O}_{q}^*\} \left(   \left\| \widehat{f}^{(q)}_t - f^{(q)}(\gamma_t) \right\|^2 
        \right)= O_P \left(  \xi_{NT} \right) $.

       \item 
       For each $q$ there is, \\
       $ \frac 1 {N^*_q T^*_q}  \sum_{i=1}^{N}\sum_{t=1}^{T} \mathbbm{1}\{(i,t)\in\mathcal{O}_{q}^*\}  \left\| \widehat{\lambda}^{(q)}_i - \widehat{\lambda}^{(q)}_{j(i)} \right\|^2 = O_P \left(  \xi_{NT} \right) $ for any matching function $(j(i),t) \in \mathcal{O}_q$ such that $g_i^{(q)} = g_{j(i)}^{(q)}$,
      and \\
      $\frac 1 {N^*_q T^*_q}  \sum_{i=1}^{N}\sum_{t=1}^{T} \mathbbm{1}\{(i,t)\in\mathcal{O}_{q}^ *\}  \left\| \widehat{f}^{\,(o)}_t - \widehat{f}^{\,(o)}_{s(t)} \right\|^2 = O_P \left(  \xi_{NT} \right) $ for any matching function $(i, s(t)) \in \mathcal{O}_q$ such that $c_t^{(q)} = c_{s(t)}^{(q)}$.

        \item 
        $\max_{k,i,t} \left| \check{X}_{it,k} \right| = O_P(1)$,
and        
        $\plim_{N,T \rightarrow \infty}  \frac 1 {{NT}} \sum_{i=1}^N \sum_{t=1}^T \check X_{it}^\prime \check X_{it} = \Omega$, where $\Omega$ is a positive definite non-random matrix.

    \end{enumerate} 
\end{assumption}

\begin{proof}[\bf{Proof of Lemma~\ref{lemma:ApproxError2}}]

    Recall from the proof of Lemma~\ref{lemma:ApproxError} the definition of $\kappa_{NT}$. Take the split sample version as follows, 
    \begin{align*}
        \kappa_{NT}^{(GS)} &:= \left( \sum_{i = 1}^N \sum_{t = 1}^T\check{X}_{it}^\prime \check{X}_{it}\right)^{-1} \sum_{i = 1}^N \sum_{t = 1}^T \check{X}_{it}^\prime \, \check \Gamma_{it} \\
        &=\left( \sum_{i = 1}^N \sum_{t = 1}^T\check{X}_{it}^\prime \check{X}_{it}\right)^{-1} \sum_{o = 1}^4 \sum_{(i,t)\in\mathcal{O}_q} [\check{X}_{it}^{(q)}]^\prime \, \check \Gamma_{it}^{(q)}.
    \end{align*}
   By Assumption~\ref{ass:ApproxErrorSplit} and the proof steps of Lemma~\ref{lemma:ApproxError}  we have that for each partition $\sum_{(i,t)\in\mathcal{O}_q} [\check{X}_{it}^{(q)}]^\prime \, \check \Gamma_{it}^{(q)} = O_P(N_qT_q\xi_{NT})$, where $N_q$ and $T_q$ are the number of $i$ and $t$, respectively, in partition $q$. Thus we have $\sum_{o = 1}^4 \sum_{(i,t)\in\mathcal{O}_q} [\check{X}_{it}^{(q)}]^\prime \, \check \Gamma_{it}^{(q)} = \sum_{o = 1}^4 O_P(N_qT_q\xi_{NT}) \leq O_P(NT\xi_{NT})$. The statement of the lemma then follows from $\sum_{i = 1}^N \sum_{t = 1}^T\check{X}_{it}^\prime \check{X}_{it} = O_P(NT)$. 
\end{proof}

\begin{proof}[\bf Proof of Lemma~\ref{lemma:CLT}]
Using the definition of $\phi^{\rm (GS)}_{NT}$
in the main text we have
\begin{align*}
   \sqrt{NT} \, \phi^{\rm (GS)}_{NT} & := 
   \widehat \Omega^{-1}  \sum_{s=1}^4  \phi^{(s)}_{NT}
\end{align*}
where
\begin{align*}
   \widehat \Omega  &:=   \frac 1 {NT} \sum_{s=1}^4
   \sum_{(i,t) \in  {\cal O}_{s}}
   \widetilde{X}_{it}^{(s)\,\prime} \widetilde{X}_{it}^{(s)} ,
   &
  \phi^{(s)}_{NT} &:= \frac 1 {\sqrt{NT}} \sum_{(i,t) \in  {\cal O}_{s}}   \widetilde{X}_{it}^{(s)\,\prime} \,  \varepsilon_{it}.
\end{align*}
By construction,
the projected regressors
 $\widetilde{X}_{it}^{(s)}$ for subpanel $s \in \{1,2,3,4\}$
 only depend on
$X=(X_{it})$, and on outcomes $Y_{it}$ (and thus error terms $\varepsilon_{it}$) that are not in that subpanel,
i.e. $(i,t) \notin {\cal O}_{s}$. Therefore, under
Assumption~\ref{ass:AsympSplit}(i), we have that for $s \in \{1,2,3,4\}$,
conditional on  $\{\widetilde{X}_{it}^{(s)} : (i,t) \in {\cal O}_{s}\}$, the $\widetilde{X}_{it}^{(s)\,\prime} \,  \varepsilon_{it}$ are mean zero and independently distributed
across all the observations $(i,t) \in {\cal O}_{s}$ in that subpanel. 
Using the regularity conditions in
Assumption~\ref{ass:AsympSplit}(ii),
for each $s \in \{1,2,3,4\}$, we can therefore
apply Lyapunov's CLT to find
\begin{align*}
   \left(\widehat \Sigma^{(s)} \right)^{-1} \,  \phi^{(s)}_{NT} \, &\Rightarrow \, {\cal N}( 0, \mathbbm{1}_K) ,
   &
    \widehat \Sigma^{(s)} &:=  \sum_{(i,t) \in  {\cal O}_{s}}
   \sigma^2_{it}
   \widetilde{X}_{it}^{(s)\,\prime} \widetilde{X}_{it}^{(s)},
\end{align*}
and the limiting distributions of 
$ \left(\widehat \Sigma^{(s)} \right)^{-1} \,  \phi^{(s)}_{NT}$
are independent across $s$.
Using that $\widehat \Sigma^{(s)}$ converges
to the constant $\Sigma^{(s)}$ we thus find that
\begin{align*}
     \sum_{s=1}^4  \phi^{(s)}_{NT}
     \Rightarrow {\cal N}\left( 0, \sum_{s=1}^4 \Sigma^{(s)} 
     \right) .
\end{align*}
Since $\widehat \Omega$ converges to $\Omega>0$,
the continuous mapping theorem then gives the statement
of the lemma.
\end{proof}

\end{document}